\documentclass[11pt,a4paper,reqno,intlimits,sumlimits]{amsart}

\usepackage{amssymb}
\usepackage{color}
\usepackage[mathscr]{euscript}
\usepackage{xspace}
\usepackage{enumerate}
\usepackage{epsfig,rotating}
\usepackage{graphicx}
\input{xy}\xyoption{all}
\usepackage{todonotes}
\usepackage{amsmath}
\usepackage{lscape}
\usepackage[round]{natbib}
\usepackage{mathtools}

\usepackage[left=1.55in,right=1.55in,top=1.6in,bottom=1.3in]{geometry}

\DeclareMathAlphabet{\mathpzc}{OT1}{pzc}{m}{it}

\usepackage[bookmarksopen,pdfstartview=FitH]{hyperref}
\hypersetup{colorlinks,%
            citecolor=blue,%
            filecolor=blue,%
            linkcolor=blue,%
            urlcolor=blue}%

\begin{document}

\theoremstyle{plain}
\newtheorem{theorem}{Theorem}[section]
\newtheorem{lemma}[theorem]{Lemma}
\newtheorem{proposition}[theorem]{Proposition}
\newtheorem{claim}[theorem]{Claim}
\newtheorem{corollary}[theorem]{Corollary}
\newtheorem{axiom}{Axiom}

\theoremstyle{definition}
\newtheorem{remark}[theorem]{Remark}
\newtheorem{note}{Note}[section]
\newtheorem{definition}[theorem]{Definition}
\newtheorem{example}[theorem]{Example}
\newtheorem*{ackn}{Acknowledgements}
\newtheorem{assumption}{Assumption}
\newtheorem{approach}{Approach}
\newtheorem{critique}{Critique}
\newtheorem{question}{Question}
\newtheorem{aim}{Aim}
\newtheorem*{assucd}{Assumption ($\mathbb{CD}$)}
\newtheorem*{asa}{Assumption ($\mathbb{A}$)}
\newtheorem*{appS}{Approximation ($\mathbb{S}$)}
\newtheorem*{appBS}{Approximation ($\mathbb{BS}$)}
\renewcommand{\theequation}{\thesection.\arabic{equation}}
\numberwithin{equation}{section}

\renewcommand{\thefigure}{\thesection.\arabic{figure}}
\numberwithin{equation}{section}

\newcommand{\Law}{\ensuremath{\mathop{\mathrm{Law}}}}
\newcommand{\loc}{{\mathrm{loc}}}

\let\SETMINUS\setminus
\renewcommand{\setminus}{\backslash}

\def\stackrelboth#1#2#3{\mathrel{\mathop{#2}\limits^{#1}_{#3}}}

\makeatletter
\def\Ddots{\mathinner{\mkern1mu\raise\p@
\vbox{\kern7\p@\hbox{.}}\mkern2mu
\raise4\p@\hbox{.}\mkern2mu\raise7\p@\hbox{.}\mkern1mu}}
\makeatother

\newcommand\llambda{{\mathchoice
     {\lambda\mkern-4.5mu{\raisebox{.4ex}{\scriptsize$\backslash$}}}
     {\lambda\mkern-4.83mu{\raisebox{.4ex}{\scriptsize$\backslash$}}}
     {\lambda\mkern-4.5mu{\raisebox{.2ex}
{\footnotesize$\scriptscriptstyle\backslash$}}}
     {\lambda\mkern-5.0mu{\raisebox{.2ex}
{\tiny$\scriptscriptstyle\backslash$}}}}}

\newcommand{\prozess}[1][L]{{\ensuremath{#1=(#1_t)_{t\in[0,T]}}}\xspace}
\newcommand{\prazess}[1][L]{{\ensuremath{#1=(#1_t)_{t\ge0}}}\xspace}
\newcommand{\pt}[1][N]{\ensuremath{\P_{#1}}\xspace}
\newcommand{\tk}[1][N]{\ensuremath{T_{#1}}\xspace}
\newcommand{\dd}[1][]{\ensuremath{\ud{#1}}\xspace}

\newcommand{\scal}[2]{\ensuremath{\langle #1, #2 \rangle}}
\newcommand{\bscal}[2]{\ensuremath{\big\langle #1, #2 \big\rangle}}
\newcommand{\set}[1]{\ensuremath{\left\{#1\right\}}}

\def\lev{L\'{e}vy\xspace}
\def\lk{L\'{e}vy--Khintchine\xspace}
\def\lib{LIBOR\xspace}
\def\mg{martingale\xspace}
\def\smmg{semimartingale\xspace}
\def\alm{affine LIBOR model\xspace}
\def\alms{affine LIBOR models\xspace}
\def\dalms{defaultable \alms}
\def\ap{affine process\xspace}
\def\aps{affine processes\xspace}

\def\half{\frac{1}{2}}

\def\F{\ensuremath{\mathcal{F}}}
\def\bD{\mathbf{D}}
\def\bF{\mathbf{F}}
\def\bG{\mathbf{G}}
\def\bH{\mathbf{H}}
\def\R{\ensuremath{\mathbb{R}}}
\def\Rp{\mathbb{R}_{\geqslant0}}
\def\Rm{\mathbb{R}_{\leqslant 0}}
\def\C{\ensuremath{\mathbb{C}}}
\def\U{\ensuremath{\mathcal{U}}}
\def\I{\mathcal{I}}
\def\N{\mathbb{N}}

\def\P{\ensuremath{\mathrm{I\kern-.2em P}}}
\def\Q{\mathbb{Q}}
\def\E{\ensuremath{\mathrm{I\kern-.2em E}}}

\def\hP{\ensuremath{\widehat{\mathrm{I\kern-.2em P}}}}
\def\hE{\ensuremath{\widehat{\mathrm{I\kern-.2em E}}}}

\def\bP{\ensuremath{\overline{\mathrm{I\kern-.2em P}}}}
\def\bE{\ensuremath{\overline{\mathrm{I\kern-.2em E}}}}

\def\bphi{\overline{\phi}}
\def\bpsi{\overline{\psi}}

\def\ott{{0\leq t\leq T}}
\def\idd{{1\le i\le d}}

\def\icc{\mathpzc{i}}
\def\ecc{\mathbf{e}_\mathpzc{i}}

\def\uk{u_{k+1}}
\def\vk{v_{k+1}}

\def\e{\mathrm{e}}
\def\ud{\ensuremath{\mathrm{d}}}
\def\dt{\ud t}
\def\ds{\ud s}
\def\dx{\ud x}
\def\dy{\ud y}
\def\dv{\ud v}
\def\dw{\ud w}
\def\dz{\ud z}
\def\dsdx{\ensuremath{(\ud s, \ud x)}}
\def\dtdx{\ensuremath{(\ud t, \ud x)}}

\def\lsnc{\ensuremath{\mathrm{LSNC-}\chi^2}}
\def\nc{\ensuremath{\mathrm{NC-}\chi^2}}

\def\red{\color{red}}
\def\blue{\color{blue}}
\newcommand{\cD}{{\mathcal{D}}}
\newcommand{\cF}{{\mathcal{F}}}
\newcommand{\cG}{{\mathcal{G}}}
\newcommand{\cH}{{\mathcal{H}}}
\newcommand{\cK}{{\mathcal{K}}}
\newcommand{\cM}{{\mathcal{M}}}
\newcommand{\cT}{{\mathcal{T}}}
\newcommand{\ha}{{\mathbb{H}}}
\newcommand{\indik}{{\mathbf{1}}}
\newcommand{\ifdefault}[1]{\ensuremath{\mathbf{1}_{\{\tau \leq #1\}}}}
\newcommand{\ifnodefault}[1]{\ensuremath{\mathbf{1}_{\{\tau > #1\}}}}

\newcommand\bovermat[2]{%
  \makebox[0pt][l]{$\smash{\overbrace{\phantom{%
    \begin{matrix}#2\end{matrix}}}^{#1}}$}#2}
\newcommand\bundermat[2]{%
  \makebox[0pt][l]{$\smash{\underbrace{\phantom{%
    \begin{matrix}#2\end{matrix}}}_{#1}}$}#2}
\makeatletter


\title[Affine LIBOR models with multiple curves]
    {Affine LIBOR models with multiple curves: theory, examples and calibration}

\author[Z. Grbac]{Zorana Grbac}
\author[A. Papapantoleon]{Antonis Papapantoleon}
\author[J. Schoenmakers]{John Schoenmakers}
\author[D. Skovmand]{David Skovmand}

\address{Laboratoire de Probabilit{\'e}s et Mod\`eles Al{\'e}atoires, 
	 Universit{\'e} Paris Diderot, 75205 Paris Cedex 13, France}
\email{grbac@math.univ-paris-diderot.fr}

\address{Institute of Mathematics, TU Berlin, Stra\ss e des 17. Juni 136,
         10623 Berlin, Germany}
\email{papapan@math.tu-berlin.de}

\address{Weierstrass Institute for Applied Analysis and Stochastics,
	 Mohrenstrasse 39, 10117 Berlin, Germany}
\email{schoenma@wias-berlin.de}

\address{Department of Finance, Copenhagen Business School, Solbjerg Plads 3, 
	 2000 Frederiksberg, Denmark}
\email{dgs.fi@cbs.dk}

\thanks{We are grateful to Fabio Mercurio and Steven Shreve for valuable 
discussions and suggestions. We also thank seminar participants at Cass 
Business School, Imperial College London, the London Mathematical Finance 
Seminar, Carnegie Mellon University and the University of Padova for their 
comments. All authors gratefully acknowledge the financial support from the DFG 
Research Center {\sc Matheon}, Projects E5 and E13.}


\begin{abstract}
We introduce a multiple curve framework that combines tract\-able dynamics and 
semi-analytic pricing formulas with positive interest rates and basis spreads. 
Negatives rates and positive spreads can also be accommodated in this framework.
The dynamics of OIS and LIBOR rates are specified following the methodology of 
the affine LIBOR models and are driven by the wide and flexible class of affine 
processes. The affine property is preserved under forward measures, which allows 
us to derive Fourier pricing formulas for caps, swaptions and basis swaptions. A 
model specification with dependent LIBOR rates is developed, that allows for an 
efficient and accurate calibration to a system of caplet prices.
\end{abstract}

\date{}\maketitle\pagestyle{myheadings}\frenchspacing

\section{Introduction}

The recent financial crisis has led to paradigm shifting events in interest
rate markets because substantial spreads have appeared between rates that used
to be closely matched; see Figure \ref{figure:SpreadHistorical} for an
illustration. We can observe, for example, that before the credit crunch the
spread between the three month LIBOR and the corresponding Overnight Indexed 
Swap (OIS) rate was non-zero, however it could be safely disregarded as 
negligible. The same is true for the three month vs six month basis swap spread.
However, since August 2007 these spreads have been evolving randomly over time,
are substantially too large to be neglected, and also depend on the tenor
length. Therefore, the assumption of a single interest rate curve that could be
used both for discounting and for generating future cash flows was seriously
challenged, which led to the introduction of the so-called \textit{multiple
curve} interest rate models.	

In the multiple curve framework, one curve is used for discounting purposes,
where the usual choice is the OIS curve, and then as many LIBOR curves as market
tenors (e.g. 1m, 3m, 6m and 1y) are built for generating future cash flows. The
difference between the OIS and each LIBOR rate is usually called basis spread or
simply basis. There are several ways of modeling the curves and different
definitions of the spread. One approach is to model the OIS and LIBOR rates
directly which leads to tractable pricing formulas, but the sign of the spread
is more difficult to control and may become negative. Another approach is to
model the OIS and the spread directly and infer the dynamics of the LIBOR; this
grants the positivity of the spread, but pricing formulas are generally less
tractable. We refer to \citet[pp.~11-12]{Mercurio_2010a} for a detailed
discussion of the advantages and disadvantages of each approach. Moreover, there
exist various definitions of the spread: an additive spread is used e.g. by
\cite{Mercurio_2010}, a multiplicative spread was proposed by
\citet{Henrard_2010}, while an instantaneous spread was used by
\cite{Andersen_Piterbarg_2010}; we refer to \cite{Mercurio_Xie_2012} for a
discussion of the merits of each definition.

The literature on multiple curve models is growing rapidly and the different
models proposed can be classified in one of the categories described above.
Moreover, depending on the modeling approach, one can also distinguish between
short rate models, Heath--Jarrow--Morton (HJM) models and LIBOR market models
(LMM) with multiple curves. The spreads appearing as modeling quantities in the
short rate and the HJM models are, by the very nature of these models,
instantaneous and given in additive form. We refer to 
\citet{Bianchetti_Morini_2013} for a detailed overview of several multiple curve
models. In the short rate framework, we mention \citet{Kenyon_2010},
\cite*{Kijima_Tanaka_Wong_2009} and \citet{Morino_Runggaldier_2014}, where the
additive short rate spread is modeled, which leads to multiplicative adjustments
for interest rate derivative prices. HJM-type models have been proposed e.g. by
\cite*{Fujii_Shimada_Takahashi_2011}, \cite*{Crepey_Grbac_Nguyen_2011}, 
\citet{Moreni_Pallavicini_2014}, \cite*{Crepey_Grbac_Ngor_Skovmand_2014} and 
\cite*{Cuchiero_Fontana_Gnoatto_2014}. The models by \citet{Mercurio_2009a}, 
\citet{Bianchetti_2010} (where an analogy with the cross-currency market has 
been exploited) and \citet{Henrard_2010} are developed in the LMM setup, while 
multiple curve extensions of the \citet{Flesaker_Hughston_1996} framework have 
been proposed in \citet{Ngyen_Seifried_2015} and 
\cite*{Crepey_Macrina_Ngyen_Skovmand_2015}. Typically, multiple curve models 
address the issue of different interest rate curves under the same currency, 
however, the paper by \cite{Fujii_Shimada_Takahashi_2011} studies a multiple 
curve model in a cross-currency setup. \citet{Filipovic_Trolle_2011} offer a 
thorough econometric analysis of the multiple curve phenomena and decompose the 
spread into a credit risk and a liquidity risk component. In recent work, 
\cite*{Gallitschke_Mueller_Seifried_2014} construct a structural model for
interbank rates, which provides an endogenous explanation for the emergence of
basis spreads.

Another important change due to the crisis is the emergence of significant 
counterparty risk in financial markets. In this paper, we consider the 
\textit{clean valuation} of interest rate derivatives meaning that we do not 
take into account the default risk of the counterparties involved in a contract. 
As explained in \cite{Crepey_Grbac_Ngor_Skovmand_2014} and 
\citet{Morino_Runggaldier_2014}, this is sufficient for calibration to market 
data which correspond to fully collateralized contracts. The price adjustments 
due to counterparty and funding risk for two particular counterparties can then 
be obtained on top of the clean prices, cf. 
\citet{Crepey_Grbac_Ngor_Skovmand_2014}, 
\cite{Crepey_Macrina_Ngyen_Skovmand_2015}, and in particular
\cite{Papapantoleon_Wardenga_2015} for computations in affine LIBOR models.

Let us also mention that there exist various other frameworks in the literature
where different curves have been modeled simultaneously, for example when
dealing with cross-currency markets (cf. e.g. \citealt{Amin_Jarrow_1991}) or
when considering credit risk (cf. e.g. the book by
\citealt{BieleckiRutkowski02}). The models in the multiple curve world often
draw inspiration from these frameworks.

\begin{figure}
 \centering
  \includegraphics[width=12.5cm]{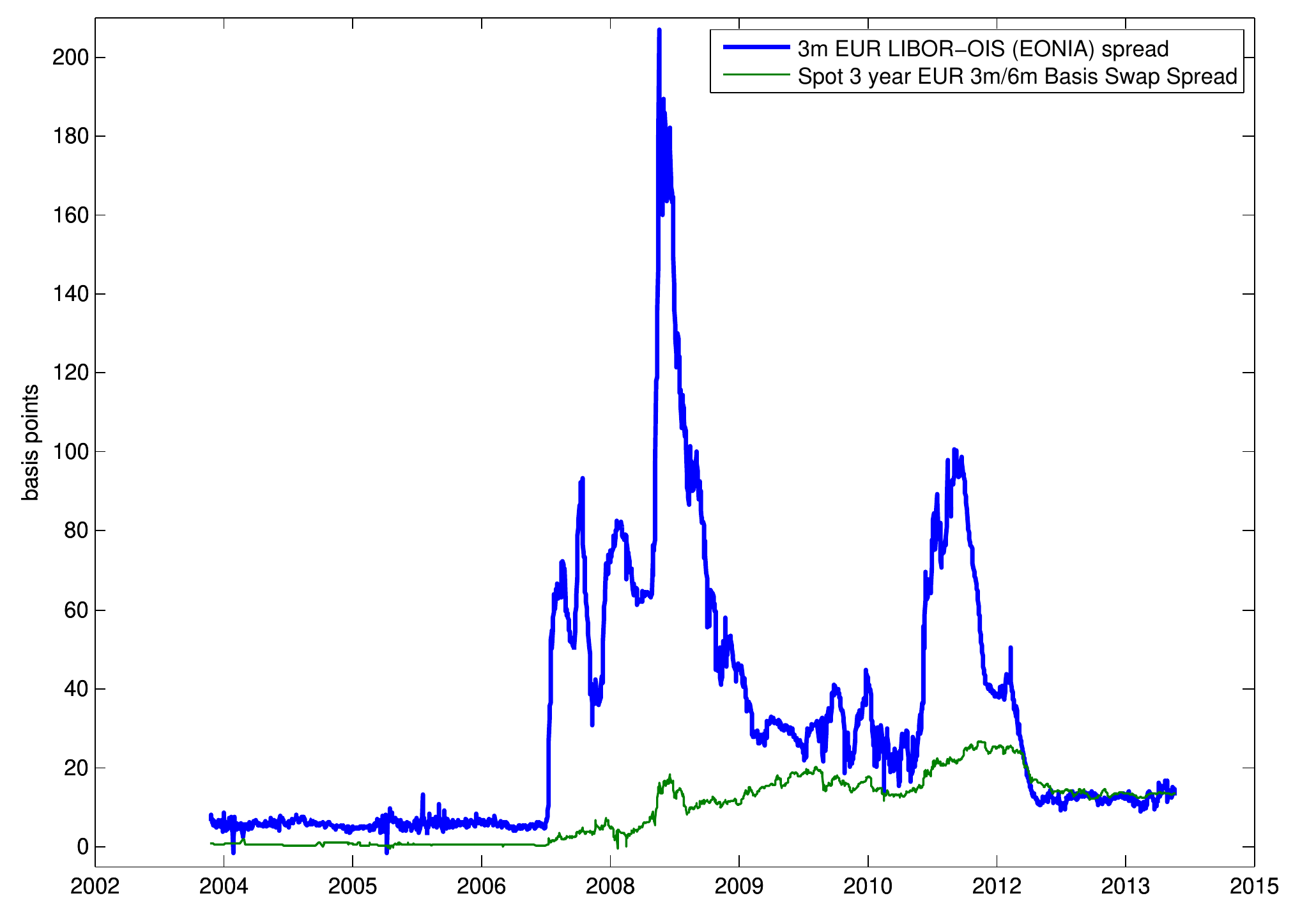}
 \vspace{-1.em}
 \caption{Spread development from January 2004 to April 2014}
 \label{figure:SpreadHistorical}
\end{figure}

The aim of this paper is to develop  a multiple curve LIBOR model that combines
tractable model dynamics and semi-analytic pricing formulas with positive
interest rates and basis spreads. The framework of the affine LIBOR models
proposed by \citet*{KellerResselPapapantoleonTeichmann09} turns out to be
tailor-made for this task, since it allows us to model directly LIBOR rates that
are greater than their OIS counterparts. In other words, the non-negativity of
spreads is automatically ensured. Simultaneously, the dynamics are driven by the
wide and flexible class of affine processes. Similar to the single curve case,
the affine property is preserved under all forward measures, which leads to
semi-analytical pricing formulas for liquid interest rate derivatives. In
particular, the pricing of caplets is as easy as in the single curve setup,
while the model structure allows to derive efficient and accurate approximations
for the pricing of swaptions and basis swaptions using a linearization of the
exercise boundary. In addition, the model offers adequate calibration results to
a system of caplet prices for various strikes and maturities.

The paper is organized as follows: in Section \ref{alm} we review the main
properties of affine processes and the construction of ordered martingales
greater than one. Section \ref{multiple-model} introduces the multiple curve
interest rate setting. The multiple curve affine LIBOR model is presented in
Section \ref{mcalm} and its main properties are discussed, in particular the
ability to produce positive rates and spreads and the analytical tractability
(i.e. the preservation of the affine property). A model that allows for negative 
interest rates and positive spreads is also presented. In Section \ref{alm-lmm} 
we study the connection between the class of affine LIBOR models and the class 
of LIBOR market models (driven by semimartingales).  Sections \ref{caps} and
\ref{val:swaptions} are devoted to the valuation of the most liquid interest
rate derivatives such as swaps, caps, swaptions and basis swaptions. In Section
\ref{calibration} we construct a multiple curve affine LIBOR model where rates
are driven by common and idiosyncratic factors and calibrate this to market
data. Moreover, we test numerically the accuracy of the swaption and basis 
swaption approximation formulas. Section \ref{epilogue} contains some concluding 
remarks and comments on future research. Finally, Appendix \ref{app-corr} 
provides an explicit formula for the terminal correlation between LIBOR rates.
\section{Affine processes}
\label{alm}

This section provides a brief review of the main properties of affine processes 
and the construction of ordered martingales greater than one. More details and
proofs can be found in \citet{KellerResselPapapantoleonTeichmann09} and the
references therein.

Let $(\Omega,\F,\bF,\P)$ denote a complete stochastic basis, where
$\bF=(\F_t)_{t\in[0,T]}$ and $T$ denotes some finite time horizon.  Consider
a stochastic process $X$ satisfying:
\begin{asa}\label{assumption-affine}
Let \prozess[X] be a conservative, time-homogeneous, stochastically
continuous Markov process with values in $D=\Rp^d$, and
$(\P_\mathrm{x})_{\mathrm{x}\in D}$ a family of probability measures on
$(\Omega,\F)$, such that $X_0=\mathrm{x}$, $\P_\mathrm{x}$-almost surely for
every $\mathrm{x}\in D$. Setting
\begin{align}
\label{eq:I_T}
\mathcal{I}_T
:= \set{u\in\R^d: \E_\mathrm{x}\big[\e^{\scal{u}{X_T}}\big] < \infty,
        \,\,\text{for all}\; \textrm{x} \in D},
\end{align}
we assume that
\begin{itemize}
\item[(i)] $0 \in \I_T^\circ$, where $\I_T^\circ$ denotes the interior of
           $\I_T$ (with respect to the topology induced by the Euclidean norm on $\R^d$);
\item[(ii)] the conditional moment generating function of $X_t$ under
	    $\P_\mathrm{x}$ has expo\-nentially-affine dependence on 
	    $\mathrm{x}$;
	    that is, there exist deterministic functions $\phi_t(u):[0,T]\times\I_T\to\R$ and
            $\psi_t(u):[0,T]\times\I_T\to\R^d$ such that
\begin{align}\label{affine-def}
\E_\mathrm{x}\big[\exp\langle u,X_t\rangle\big]
 = \exp\big( \phi_t(u) + \langle\psi_t(u),\mathrm{x}\rangle \big),
\end{align}
for all $(t,u,\mathrm{x}) \in [0,T] \times \I_T \times D$.
\end{itemize}
\end{asa}
\noindent Here $\langle\cdot,\cdot\rangle$ denotes the inner product on $\R^d$
and $\E_\mathrm{x}$ the expectation with respect to $\P_\mathrm{x}$. Moreover, 
it holds that $\I_T\subseteq \I_t$ for $t\leq T$, cf. 
\citet[Theorem 2.14]{KellerResselMayerhofer15}. In other words, if $u \in \R^d$ 
is such that $ \E_\mathrm{x}\big[\e^{\scal{u}{X_T}}\big] < \infty$, then $ 
\E_\mathrm{x}\big[\e^{\scal{u}{X_t}}\big] < \infty$ for every $t\leq T$.

The functions $\phi$ and $\psi$ satisfy the following system of ODEs, known as
\textit{generalized Riccati equations}
\begin{subequations}\label{Riccati}
\begin{align}
\frac{\partial}{\partial t}\phi_t(u)
 &= F(\psi_t(u)),  \qquad \phi_0(u)=0, \label{Ric-1}\\
\frac{\partial}{\partial t}\psi_t(u)
 &= R(\psi_t(u)),  \qquad \psi_0(u)=u, \label{Ric-2}
\end{align}
\end{subequations}
for $(t,u)\in [0,T] \times \I_T$. The functions $F$ and $R$ are of \lk form:
\begin{subequations}\label{F-R-def}
\begin{align}
F(u) &= \langle b,u\rangle +
     \int_D\big(\e^{\langle\xi,u\rangle}-1\big)m(\ud \xi),\\
R_i(u) &= \langle \beta_i,u\rangle
       + \Big\langle\frac{\alpha_i}2u,u\Big\rangle
       + \int_D\big(\e^{\langle\xi,u\rangle}-1-\langle
          u,h_i(\xi)\rangle\big)\mu_i(\ud \xi),
\end{align}
\end{subequations}
where $(b,m,\alpha_i,\beta_i,\mu_i)_{1\le i\le d}$ are \textit{admissible
parameters} and $h_i:\Rp^d\to\R^d$ are suitable truncation functions.
The functions $\phi$ and $\psi$ also satisfy the semi-flow equations
\begin{subequations}\label{flow}
\begin{align}
\phi_{t+s}(u) &= \phi_{t}(u)+\phi_{s}(\psi_t(u))\\
\psi_{t+s}(u) &= \psi_{s}(\psi_{t}(u))
\end{align}
\end{subequations}
for all $0 \le t+s \le T$ and $u \in \I_T$, with initial condition
\begin{align}\label{phi-psi-0}
 \phi_0(u)=0
  \quad\text{ and }\quad
 \psi_0(u)=u.
\end{align}
We refer to \citet*{DuffieFilipovicSchachermayer03} for all the details.

The following definition will be used in the sequel, where 
$\mathbf{1}:=(1,1,\dots,1)$. 

\begin{definition}\label{def-gamma}
Let $X$ be a process satisfying Assumption $(\mathbb{A})$. Define
\begin{align}
\gamma_X \,
 := \sup_{u\in\I_T\cap\R^d_{>0}} \E_{\mathbf{1}}\big[\e^{\scal{u}{X_T}}\big].
\end{align}
\end{definition}

\noindent The quantity $\gamma_X$ measures the ability of an affine process to 
fit the initial term structure of interest rates and equals infinity for several 
models used in mathematical finance, such as the CIR process and OU models 
driven by subordinators; cf.
\citet[\S8]{KellerResselPapapantoleonTeichmann09}.

An essential ingredient in affine LIBOR models is the construction of 
parame\-trized martingales which are greater than or equal to one and increasing 
in this parameter (see also \citealt{Papapantoleon10b}).

\begin{lemma}\label{ord-mart}
Consider an affine process $X$ satisfying Assumption ($\mathbb{A}$) and let
$u\in\mathcal{I}_T\cap\Rp^d$. Then the process \prozess[M^u] with
\begin{align}
M^u_t &= \exp\big( \phi_{T-t}(u) + \scal{\psi_{T-t}(u)}{X_t} \big),
\end{align}
is a martingale, greater than or equal to one, and the mapping $u\mapsto M^u_t$ 
is increasing, for every $t\in[0,T]$.
\end{lemma}

\begin{proof}
Consider the random variable $Y^u_T:=\e^{\scal{u}{X_T}}$. Since
$u\in\mathcal{I}_T\cap\Rp^d$ we have that $Y^u_T$ is greater than one and
integrable. Then, from the Markov property of $X$, \eqref{affine-def} and the
tower property of conditional expectations we deduce that
\begin{align}\label{Pn-martingales}
M^u_t = \E\big[\e^{\scal{u}{X_T}}|\F_t\big]
           = \exp\big( \phi_{T-t}(u) + \scal{\psi_{T-t}(u)}{X_t} \big),
\end{align}
is a martingale. Moreover, it is obvious that $M^u_t\ge1$ for all $t\in[0,T]$,
while the ordering
\begin{align}\label{M-order}
 u\le v \,\,\Longrightarrow\,\, M_t^u\le M_t^v
\qquad \forall t\in[0,T],
\end{align}
follows from the ordering of $Y_T^u$ and the representation
$M^u_t=\E[Y_T^u|\F_t]$.
\end{proof}
\section{A multiple curve LIBOR setting}
\label{multiple-model}

We begin by introducing the notation and the main concepts of multiple curve 
\lib models. We will follow the approach introduced in \citet{Mercurio_2010}, 
which has become the industry standard in the meantime.

The fact that LIBOR-OIS spreads are now tenor-dependent means that we cannot 
work with a single tenor structure any longer. Hence, we start with a discrete, 
equidistant time structure $\cT=\{0=T_0<T_1<\cdots<T_N\}$, where $T_k$, 
$k\in\mathcal{K}:=\{1,\ldots,N\}$, denote the maturities of the assets traded 
in the market. Next, we consider different subsets of $\cT$ with equidistant 
time points, i.e. different tenor structures 
$\cT^x=\{0=T_0^x<T_1^x<\cdots<T_{N^x}^x\}$, where 
$x\in\mathcal{X}:=\{x_1,x_2,\ldots,x_n\}$ is a label that indicates the tenor 
structure. Typically, we have $\mathcal{X}=\{1,3,6,12\}$ months. We denote the 
tenor length by $\delta_x=T_k^x-T_{k-1}^x$, for every $x\in\mathcal{X}$. Let 
$\mathcal{K}^x:=\{1,2,\dots,N^x\}$ denote the collection of all subscripts 
related to the tenor structure $\cT^{x}$. We assume that $\cT^x\subseteq\cT$ 
and $T_{N^{x}}^{x}=T_N$ for all $x\in\mathcal{X}$. A graphical illustration of 
a possible relation between the different tenor structures appears in Figure 
\ref{Fig:tenors}.

\begin{example}
\label{ex:tenors-1}
A natural construction of tenor structures is the following: Let 
$\cT=\{0=T_0<T_1<\cdots<T_N\}$ denote a discrete time structure, where 
$T_k=k\delta$ for $k=1,\ldots,N$ and $\delta>0$. Let 
$\mathcal{X}=\{1=x_1,x_2,\ldots,x_n\}\subset\N$, where we assume that $x_k$ is a 
divisor of $N$ for all $k=1,\dots,n$. Next, set for every $x\in\mathcal{X}$
\begin{align*}
T_k^x = k\cdot\delta\cdot x=:k\delta_x,
 \quad\text{ for }\, k=1,\dots,N^x:=N/x,
\end{align*}
where obviously $T_k^x=T_{kx}$. Then, we can consider different subsets of 
$\cT$, i.e. different tenor structures 
$\cT^x=\{0=T_0^x<T_1^x<\cdots<T_{N^x}^x\}$, which satisfy by construction 
$\cT^x\subset\cT^{x_1}=\cT$ and also $T_{N^{x}}^{x}=N^x\cdot\delta\cdot x=T_N$, 
for all $x\in\mathcal{X}$.
\end{example}

\begin{figure}
\begin{center}
\setlength{\unitlength}{0.625cm}
\begin{picture}(18.0,7.)(0,-5.75)
 \thicklines
\put(0, 0){\line(1, 0){18}}
\put(0, -0.2){\line(0, 1){0.4}}
\put(-0.1,-0.8){$T_0$}
\put(1.5, -0.2){\line(0, 1){0.4}}
\put(1.4,-0.8){$T_1$}
\put(3.0, -0.2){\line(0, 1){0.4}}
\put(2.9,-0.8){$T_2$}
\put(4.5, -0.2){\line(0, 1){0.4}}
\put(4.4,-0.8){$T_3$}
\put(6.0, -0.2){\line(0, 1){0.4}}
\put(5.9,-0.8){$T_4$}
\put(7.5, -0.2){\line(0, 1){0.4}}
\put(7.4,-0.8){$T_5$}
\put(9.0, -0.2){\line(0, 1){0.4}}
\put(8.9,-0.8){$T_6$}
\put(12.5,-0.8){$\cdots$}
\put(16.5, -0.2){\line(0, 1){0.4}}
\put(16.4,-0.8){$T_{n-1}$}
\put(18, -0.2){\line(0, 1){0.4}}
\put(17.9,-0.8){$T_N$}
\put(0, -2.25){\line(1, 0){18}}
\put(0, -2.45){\line(0, 1){0.4}}
\put(-0.1,-3.05){$T_0$}
\put(4.5, -2.45){\line(0, 1){0.4}}
\put(4.4,-3.05){$T_1^{x_1}$}
\put(9.0, -2.45){\line(0, 1){0.4}}
\put(8.9,-3.05){$T_2^{x_1}$}
\put(12.5,-3.05){$\cdots$}
\put(18, -2.45){\line(0, 1){0.4}}
\put(17.9,-3.05){$T_{N^{x_1}}^{x_1}$}
\put(0, -4.5){\line(1, 0){18}}
\put(0, -4.7){\line(0, 1){0.4}}
\put(-0.1,-5.3){$T_0$}
\put(9.0, -4.7){\line(0, 1){0.4}}
\put(8.9,-5.3){$T_1^{x_2}$}
\put(12.5,-5.3){$\cdots$}
\put(18, -4.7){\line(0, 1){0.4}}
\put(17.9,-5.3){$T_{N^{x_2}}^{x_2}$}
\end{picture}
\end{center}
\vspace{-1.em}\caption{Illustration of different tenor structures.}
\label{Fig:tenors}
\end{figure}
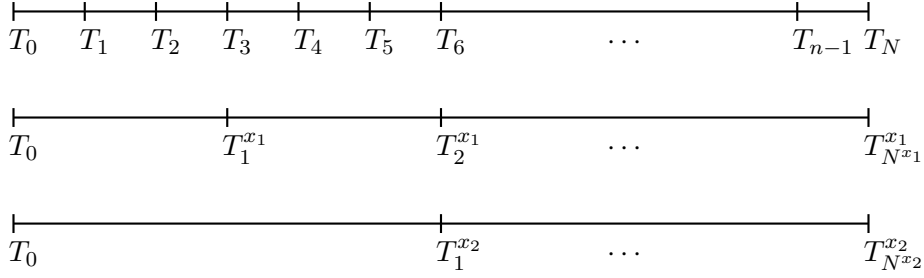

We consider the OIS curve as discount curve, following the standard market 
practice of fully collateralized contracts. The market prices for caps and 
swaptions considered in the sequel for model calibration are indeed quoted 
under the assumption of full collateralization. A detailed discussion on the 
choice of the discount curve in the multiple curve setting can be found e.g. in 
\citet{Mercurio_2010} and in \citet{Hull_White_2013}. The discount factors 
$B(0,T)$ are stripped from market OIS rates and defined for every possible 
maturity $T\in\cT$ via
\begin{align*}
T \mapsto B(0,T) = B^{OIS}(0,T).
\end{align*}
We denote by $B(t,T)$ the discount factor, i.e. the price of a zero coupon bond, 
at time $t$ for maturity $T$, which is assumed to coincide with the 
corresponding OIS-based zero coupon bond for the same maturity.

We also assume that all our modeling objects live on a complete stochastic 
basis $(\Omega,\F,\bF,\P_N)$, where $\P_N$ denotes the terminal forward measure, 
i.e. the martingale measure associated with the numeraire $B(\cdot,T_N)$. The 
corresponding expectation is denoted by $\E_N$. Then, we introduce forward 
measures $\P_k^x$ associated to the numeraire $B(\cdot,T_k^x)$ 
for every pair $(x,k)$ with $x\in\mathcal{X}$ and $k\in\mathcal{K}^x$. The 
corresponding expectation is denoted by $\E_k^x$. The forward measures $\P_k^x$ 
are absolutely continuous with respect to $\P_N$, and defined in the usual way, 
i.e. via the Radon--Nikodym density
\begin{align}
\frac{\ud \P_k^x}{\ud \P_N}=
 \frac{B(0,T_N)}{B(0,T_k^x)} \frac{1}{B(T_k^x,T_N)}.
\end{align}

\begin{remark}\label{r:tenors}
Since $\cT^x\subseteq\cT$ there exists an $l\in\mathcal{K}$ and a 
$k\in\mathcal{K}^x$ such that $T_l=T_k^x$. Therefore, the corresponding 
numeraires and forward measures coincide, i.e. $B(\cdot,T_l)=B(\cdot,T_k^x)$ and 
$\P_l=\P_k^x$. See again Figure \ref{Fig:tenors}.
\end{remark}

Next, we define the two rates that are the main modeling objects in the multiple 
curve LIBOR setting: the forward OIS rate and the forward LIBOR rate. We also 
define the additive and the multiplicative spread between these two rates. Let 
us denote by $L(T_{k-1}^x,T_k^x)$ the \textit{spot LIBOR rate} at time 
$T_{k-1}^x$ for the time interval $[T_{k-1}^x,T_k^x]$, which is an 
$\F_{T^x_{k-1}}$-measurable random variable on the given stochastic basis.

\begin{definition}
The time-$t$ \textit{forward OIS rate} for the time interval $[T_{k-1}^x,T_k^x]$
is defined by
\begin{align}\label{OIS-defin}
F_k^x(t)
 := \frac{1}{\delta_x}
    \left( \frac{B(t,T_{k-1}^x)}{B(t, T_k^x)} -1 \right).
\end{align}
\end{definition}

\begin{definition}
The time-$t$ \textit{forward LIBOR rate} for the time interval 
$[T_{k-1}^x,T_k^x]$ is defined by
\begin{align}\label{LIBOR-defin}
L_k^x(t)
  = \E_k^{x} \big[L(T_{k-1}^x, T_k^x) | \cF_t\big].
\end{align}
\end{definition}

The forward LIBOR rate is the fixed rate that should be exchanged for the 
future spot LIBOR rate so that the forward rate agreement has zero initial 
value. Hence, this rate reflects the market expectations about the value of 
the future spot LIBOR rate. Notice that at time $t=T_{k-1}^x$ we have that
\begin{align}\label{eq:LIBOR-LIB-connection}
L_k^x(T_{k-1}^x)
 = \E_k^{x} \big[L(T_{k-1}^x, T_k^x) | \cF_{T_{k-1}^x}\big]
 = L(T_{k-1}^x, T_k^x),
\end{align}
i.e. this rate coincides with the spot LIBOR rate at the corresponding tenor 
dates.

\begin{remark}\label{LIBOR-neq-ZCB}
In the single curve setup, \eqref{OIS-defin} is the definition of the forward 
LIBOR rate. However, in the multiple curve setup we have that
$$
L(T_{k-1}^x, T_k^x) \neq \frac{1}{\delta_x}
 \left(\frac{1}{B(T_{k-1}^x, T_k^x)} -1\right),
$$
hence the OIS and the LIBOR rates are no longer equal.
\end{remark}

\begin{definition}
The \textit{spread} between the LIBOR and the OIS  rate is defined by
\begin{align}\label{spread}
S_k^x(t):= L_k^x(t) - F_k^x(t).
\end{align}
\end{definition}

Let us also provide an alternative definition of the spread based on a 
multiplicative, instead of an additive, decomposition.

\begin{definition}
The \textit{multiplicative spread} between the LIBOR and the OIS 
rate is defined by
\begin{align}\label{mlpl-spread}
1 + \delta_x R_k^x(t) 
:= \frac{1 + \delta_x L_k^x(t)}{1 + \delta_x F_k^x(t)}.
\end{align}
\end{definition}

A model for the dynamic evolution of the OIS and LIBOR rates, and thus also of 
their spread, should satisfy certain conditions which stem from economic 
reasoning, arbitrage requirements and their respective definitions. These are, 
in general, consistent with market observations. We formulate them below as 
model requirements:
\begin{list}{\textbf{(M1)}}{}
\item $F_k^x(t)\ge0$ and $F_k^x\in\mathcal{M}(\P_k^x)$, for all
      $x\in\mathcal{X}$, $k\in\mathcal{K}^x$, $t\in[0,T_{k-1}^x]$.
\end{list}
\begin{list}{\textbf{(M2)}}{}
\item $L_k^x(t)\ge0$ and $L_k^x\in\mathcal{M}(\P_k^x)$, for all
      $x\in\mathcal{X}$, $k\in\mathcal{K}^x$, $t\in[0,T_{k-1}^x]$.
\end{list}
\begin{list}{\textbf{(M3)}}{}
\item $S_k^x(t)\ge0$ and $S_k^x\in\mathcal{M}(\P_k^x)$, for all
      $x\in\mathcal{X}$, $k\in\mathcal{K}^x$, $t\in[0,T_{k-1}^x]$.
\end{list}
Here $\mathcal{M}(\P_k^x)$ denotes the set of $\P_k^x$-martingales.

\begin{remark}
If the additive spread is positive the multiplicative spread is also positive 
and vice versa.
\end{remark}
\section{The multiple curve affine \lib model}
\label{mcalm}

We describe next the affine \lib model for the multiple curve interest rate 
setting and analyze its main properties. In particular, we show that this model 
produces positive rates and spreads, i.e. it satisfies the modeling 
requirements (M1)--(M3) and is analytically tractable. In this framework, OIS 
and LIBOR rates are modeled in the spirit of the \alm introduced by 
\citet{KellerResselPapapantoleonTeichmann09}.

Let $X$ be an affine process defined on $(\Omega,\F,\bF,\P_N)$, satisfying
Assumption $(\mathbb{A})$ and starting at the canonical value $\mathbf{1}$.
Consider a fixed $x\in\mathcal{X}$ and the associated tenor structure
$\mathcal{T}^x$. We construct two families of parametrized martingales following
Lemma \ref{ord-mart}: take two sequences of vectors
$(u_k^x)_{k\in\mathcal{K}^x}$ and $(v_k^x)_{k\in\mathcal{K}^x}$, and define the
$\P_N$-martingales $M^{u^x_k}$ and $M^{v^x_k}$ via
\begin{equation}\label{eq:Mu}
M^{u^x_k}_t
 = \exp\big( \phi_{T_N-t}(u^x_k) + \scal{\psi_{T_N-t}(u^x_k)}{X_t} \big),
\end{equation}
and
\begin{equation}\label{eq:Mv}
M^{v^x_k}_t
 = \exp\big( \phi_{T_N-t}(v^x_k) + \scal{\psi_{T_N-t}(v^x_k)}{X_t} \big).
\end{equation}
The \textit{multiple curve affine \lib model} postulates that the OIS and the
LIBOR rates associated with the $x$-tenor evolve according to
\begin{align}
1 + \delta_x F_k^x(t) & =  \frac{M_t^{u_{k-1}^x}}{M_t^{u_k^x}} 
\quad\text{ and }\quad
\label{eq:LIBOR-rate}
1 + \delta_x L_k^x(t)  = \frac{M^{v_{k-1}^x}_t}{M^{u_{k}^x}_t},
\end{align}
for every $k=2,\ldots,N_x$ and $t\in[0,T_{k-1}^x]$.

In the following three propositions, we show how to construct a multiple curve
affine \lib model from any given initial term structure of OIS and LIBOR rates.

\begin{proposition}\label{initial-fit-multiple-curve-OIS}
Consider the time structure $\mathcal{T}$, let $B(0,T_l)$, $l \in \mathcal{K}$,
be the initial term structure of non-negative OIS discount factors and assume
that
\begin{equation*}
B(0,T_1) \geq \cdots \geq B(0,T_{N}).
\end{equation*}
Then the following statements hold:
\begin{enumerate}
\item If $\gamma_X>B(0,T_1)/B(0,T_{N})$, then there exists a decreasing sequence
 $u_1\ge u_2\ge\dots\ge u_{N} =0$ in $\I_T\cap\Rp^d$, such that
 \begin{equation}\label{eq:initial-OIS-fit}
  M_0^{u_l} = \frac{B(0,T_l)}{B(0,T_N)}
   \qquad \text{for all}\;\, l \in \mathcal{K}.
\end{equation}
In particular, if $\gamma_X=\infty$, the multiple curve affine \lib model can
fit any initial term structure of OIS  rates.
\item If $X$ is one-dimensional, the sequence
 $(u_l)_{l\in\mathcal{K}}$  is unique.
\item If all initial OIS rates are positive, the sequence
 $(u_l)_{l\in\mathcal{K}}$ is strictly decreasing.
\end{enumerate}
\end{proposition}
\begin{proof}
See Proposition 6.1 in \citet{KellerResselPapapantoleonTeichmann09}.
\end{proof}

After fitting the initial term structure of OIS discount factors, we want to
fit the initial term structure of LIBOR rates, which is now tenor-dependent.
Thus, for each $k\in\mathcal{K}^x$, we set
\begin{align}
\label{eq:u}
u_k^x & := u_l,
\end{align}
where $l\in\mathcal{K}$ is such that $T_l=T_k^x$; see Remark \ref{r:tenors}. In
general, we have that $l=kT_1^x/T_1$, while in the setting of Example
\ref{ex:tenors-1} we simply have $l= kx$, i.e. $u_k^x=u_{kx}$.

\begin{proposition}\label{initial-fit-multiple-curve}
Consider the setting of Proposition \ref{initial-fit-multiple-curve-OIS}, the
fixed $x\in\mathcal{X}$ and the corresponding tenor structure $\mathcal{T}^x$.
Let $L^x_k(0)$, $k\in\mathcal{K}^x$, be the initial term structure of
non-negative LIBOR rates and assume that for every $k\in\mathcal{K}^x$
\begin{equation}\label{eq:initial-LIBOR}
 L^x_{k}(0) \geq
 \frac{1}{\delta_x}\left(\frac{B(0,T_{k-1}^x)}{B(0,T_{k}^x)}-1\right)
  = F_k^x(0).
\end{equation}
The following statements hold:
\begin{enumerate}
\item If
$\gamma_X>\max_{k\in\mathcal{K}^x}(1+\delta_xL_k^x(0))B(0,T_k^x)/B(0,T_N^x)$,
then there exists a sequence $v_1^x,v_2^x,\dots,v_{N^x}^x=0$ in
$\I_T\cap\Rp^d$, such that $v_k^x
\geq u_k^x$ and
 \begin{equation}\label{eq:initial-LIBOR-fit}
  M_0^{v_{k}^x} = (1 + \delta_x L_{k+1}^x(0)) M_0^{u_{k+1}^x},
  \quad \text{for all}\;\, k \in \mathcal{K}^x\setminus\{N^x\}.
\end{equation}
 In particular, if $\gamma_X=\infty$, then the multiple curve affine \lib model
 can fit any initial term structure of LIBOR rates.
\item If $X$ is one-dimensional, the sequence
 $(v_k^x)_{k\in\mathcal{K}^x}$ is unique.
\item  If all initial spreads are positive, then $v_k^x > u_k^x$, for all
$k\in\mathcal{K}^x\setminus\{N^x\}$.
\end{enumerate}
\end{proposition}

\begin{proof}
Similarly to the previous proposition, by fitting the initial LIBOR rates we
obtain a sequence $(v_k^x)_{k \in \mathcal{K}^x}$ which satisfies
(1)--(3). The inequality $v_k^x\ge u_k^x$ follows directly from
\eqref{eq:initial-LIBOR}.
\end{proof}

\begin{proposition}
Consider the setting of the previous propositions. Then we have:
\begin{enumerate}
\item $F_k^x$ and $L_k^x$ are $\P_k^x$-martingales, for every
$k\in\mathcal{K}^x$.
\item $L_k^x(t) \geq F_k^x(t) \geq 0$, for every $k\in\mathcal{K}^x$,
$t\in[0,T_{k-1}^x]$.
\end{enumerate}
\end{proposition}

\begin{proof}
Since $M^{u_k^x}$ and $M^{v_k^x}$ are $\P_N$-martingales and the density process
relating the measures $\P_N$ and $\P_k^x$ is provided by
\begin{align}\label{Pkx-densities}
\frac{\ud \P_k^x}{\ud \P_N}\Big|_{\cF_t}
 = \frac{B(0, T_N)}{B(0, T_k^x)} \frac{B(t,T_k^x)}{B(t,T_N)}
 = \frac{M^{u_k^x}_t}{M^{u_k^x}_0},
\end{align}
we get from \eqref{eq:LIBOR-rate} that
\begin{align}
1 + \delta_x F_k^x \in \cM(\P_k^{x})
\,\text{ because }\,
(1 + \delta_xF_k^x) M^{u_{k}^x} = M^{u_{k-1}^x}  \in \cM(\P_N).
\end{align}
Similarly,
\begin{align}
1 + \delta_x L_k^x \in \cM(\P_k^{x})
\,\text{ because }\,
(1 + \delta_xL_k^x) M^{u_{k}^x} = M^{v_{k-1}^x}  \in \cM(\P_N).
\end{align}

The monotonicity of the sequence $(u_k^x)$ together with \eqref{M-order} yields
that $M^{u^x_{k-1}}\geq M^{u^x_{k}}$. Moreover, from the inequality $v_k^x\geq
u_k^x$ together with \eqref{M-order} again, it follows that $M^{v^x_{k}}\geq
M^{u^x_{k}} $, for all $k\in\mathcal{K}^x$. Hence,
$$
1 + \delta_x L_k^x \geq 1 + \delta_x F_k^x \geq 1.
$$
Therefore, the OIS rates, the LIBOR rates and the corresponding spreads
are non-negative $\P_k^x$-martingales.
\end{proof}

\begin{remark}
\label{r:choice-of-u-v}
The above propositions provide the theoretical construction of affine LIBOR 
models with multiple curves, given initial term structures of OIS bond prices 
$B(0, T_k^x)$ and LIBOR rates $L_k^x(0)$, for any $x \in \mathcal{X}$ and $k\in 
\mathcal{K}^x$. The initial term structures determine the sequences $(u_k^x)$ 
and $(v_k^x)$, but not in a unique way, as soon as the dimension of the driving 
process is strictly greater than one, which will typically be the case in 
applications. This provides plenty of freedom in the implementation of the 
model. For example, setting some components of the vectors  $(u_k^x)$ and 
$(v_k^x)$ equal to zero allows to exclude the corresponding components of the 
driving processes and thus decide which components of the driving process $X$ 
will affect the OIS rates, respectively the LIBOR rates. Moreover, if the 
components of $X$ are assumed mutually independent, one can create a factor 
model with common and idiosyncratic components for each OIS and LIBOR rate, as 
well as various other specific structures. In Section 
\ref{subsection:SingleTenor} we present more details on this issue, see Remark 
\ref{obsLFS} in particular.
\end{remark}


\begin{remark}\label{rem:relation-u-v}
Let us now look more closely at the relationship between the sequences $(v_k^x)$ 
and $(u_k^x)$. Propositions \ref{initial-fit-multiple-curve-OIS} and 
\ref{initial-fit-multiple-curve} imply that $u_{k-1}^x\geq u_{k}^x$ and $v_k^x
\geq u_k^x$ for all $k\in\mathcal{K}^x$. However, we do not know the ordering
of $v_k^x$ and $u_{k-1}^x$, or whether the sequence $(v_k^x)$ is monotone or
not. The market data for LIBOR spreads indicate that in a `normal' market 
situation $v_k^x\in[u_{k}^x,u_{k-1}^x]$. More precisely, on the one hand, we 
have $v_k^x\geq u_k^x$ because the LIBOR spreads are nonnegative. On the other 
hand, if $v_k^x>u_{k-1}^x$, then the LIBOR rate would be more than two times 
higher than the OIS rate spanning an interval twice as long, starting at the 
same date. This contradicts normal market behavior, hence 
$v_k^x\in[u_{k}^x,u_{k-1}^x]$ and consequently the sequence $(v_k^x)$ will also 
be decreasing. This ordering of the parameters $(v_k^x)$ and $(u_k^x)$ is 
illustrated in Figure \ref{Fig:u-and-v} (top graph). However, the `normal' 
market situation alternates with an `extreme' situation, where the spread is 
higher than the OIS rate. In the bottom graph of Figure \ref{Fig:u-and-v} we 
plot another possible ordering of the parameters $(v_k^x)$ and $(u_k^x)$ 
corresponding to such a case of very high spreads. Intuitively speaking, the 
value of the corresponding model spread depends on the distance between the 
parameters $(v_k^x)$ and $(u_k^x)$, although in a non-linear fashion. 
\end{remark}

\begin{figure}
\begin{center}
\setlength{\unitlength}{0.625cm}
\begin{picture}(19.0,2.)(0,-1.75)
 \thicklines
\put(0, 0){\vector(1, 0){19}}
\put(0, -0.2){\line(0, 1){0.4}}
\put(-0.1,-0.8){$0=u^x_{N^x}$}
\put(3.0, -0.2){\line(0, 1){0.4}}
\put(2.9,-0.8){$u_{{N^x}-1}^x$}
\put(4.7, -0.2){\line(0, 1){0.4}}
\put(4.6,-0.8){$v_{{N^x}-1}^x$}
\put(6.7, -0.8){\dots}
\put(8.5, -0.2){\line(0, 1){0.4}}
\put(8.4,-0.8){$u_{k}^x$}
\put(10.8, -0.2){\line(0, 1){0.4}}
\put(10.7,-0.8){$v_{k}^x$}
\put(12.0, -0.2){\line(0, 1){0.4}}
\put(11.9,-0.8){$u_{k-1}^x$}
\put(13.9, -0.8){\dots}
\put(16.2, -0.2){\line(0, 1){0.4}}
\put(16.1,-0.8){$u_1^x$}
\put(18, -0.2){\line(0, 1){0.4}}
\put(17.9,-0.8){$v_1^x$}
\put(0, -2){\vector(1, 0){19}}
\put(0, -2.2){\line(0, 1){0.4}}
\put(-0.1,-2.8){$0=u^x_{N^x}$}
\put(2.5, -2.2){\line(0, 1){0.4}}
\put(2.4,-2.8){$u_{{N^x}-1}^x$}
\put(4.3, -2.8){\dots}
\put(5.7, -2.2){\line(0, 1){0.4}}
\put(5.6,-2.8){$u_k^x$}
\put(7.0, -2.8){\dots}
\put(8.5, -2.2){\line(0, 1){0.4}}
\put(8.4,-2.8){$u_1^x$}
\put(11.8, -2.2){\line(0, 1){0.4}}
\put(11.7,-2.8){$v_1^x$}
\put(13.3,-2.8){\dots}
\put(15.0, -2.2){\line(0, 1){0.4}}
\put(14.9, -2.8){$v_{{N^x}-1}^x$}
\put(16.7,-2.8){\dots}
\put(18, -2.2){\line(0, 1){0.4}}
\put(17.9,-2.8){$v_k^x$}
\end{picture}
\end{center}
\vspace{1.75em}
\caption{Two possible orderings of $(u_k^x)$ and $(v_k^x)$.}
\label{Fig:u-and-v}
\end{figure}
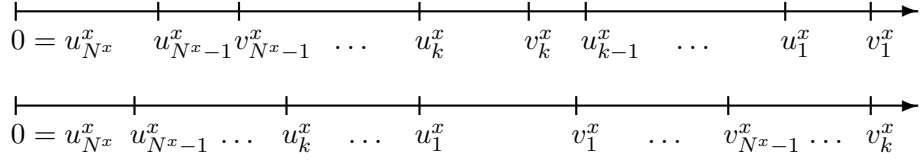

The next result concerns an important property of the multiple curve \alm, 
namely its analytical tractability in the sense that the model structure is 
preserved under different forward measures. More precisely, the process $X$ 
remains affine under any forward measure, although its `characteristics' become 
time-dependent. We refer to \citet{Filipovic05} for time-inhomogeneous affine 
processes. This property plays a crucial role in the derivation of tractable 
pricing formulas for interest rate derivatives in the forthcoming sections, 
since it entails that the law of any collection of LIBOR rates is known under 
any forward measure. The result below is presented in 
\citet[cf. eq. (6.14) and its proof]{KellerResselPapapantoleonTeichmann09}, 
nevertheless we include a short proof here for completeness. In Section 
\ref{alm-lmm} we also provide an alternative proof for the case when $X$ is an 
affine diffusion. 

\begin{proposition}\label{X-Pkx-characteristics}
The process $X$ is a time-inhomogeneous affine process under the measure
$\P_k^x$, for every $x\in\mathcal{X}$ and $k\in\mathcal{K}^x$. In particular
\begin{align}
\label{X-Pxk-characteristics-1}
\E_{k}^x \big[ \e^{\scal{w}{X_t}} \big]
 &= \exp\left( \phi_t^{k,x}(w) + \scal{\psi_t^{k,x}(w)}{X_0} \right),
\end{align}
where
\begin{subequations}
\begin{align}
\phi_t^{k,x}(w) &:= \phi_t\big(\psi_{T_N-t}(u_k^x)+w\big)
                  - \phi_t\big(\psi_{T_N-t}(u_k^x)\big),\\
\psi_t^{k,x}(w) &:= \psi_t\big(\psi_{T_N-t}(u_k^x)+w\big)
                  - \psi_t\big(\psi_{T_N-t}(u_k^x)\big),
\end{align}
\end{subequations}
for every $w \in \mathcal{I}^{k,x}$ with
\begin{align}\label{eq:I-kx}
\mathcal{I}^{k,x} := \set{w\in\R^d: \psi_{T_N-t}(u_k^x) + w \in \mathcal{I}_T}.
\end{align}
\end{proposition}

\begin{proof}
Using the density process between the forward measures, see
\eqref{Pkx-densities}, we have that
\begin{align}\label{X-Pkx-computation}
\E_k^x \big[ \e^{\scal{w}{X_t}} \big|\F_s \big]
 &= \E_N \Big[ \e^{\scal{w}{X_t}} M_t^{u_k^x}/M_s^{u_k^x} \big|\F_s \Big]
\nonumber\\
 &= \E_N \Big[ \exp\big( \phi_{T_N-t}(u_k^x)
         + \scal{\psi_{T_N-t}(u_k^x)+w}{X_t} \big)
                                \big|\F_s \Big] /M_s^{u_k^x} \nonumber\\
 &= \exp\Big( \phi_{T_N-t}(u_k^x) - \phi_{T_N-s}(u_k^x)
                     + \phi_{t-s}(\psi_{T_N-t}(u_k^x)+w)\Big) \nonumber\\
 &\quad\times
       \exp\Big\langle \psi_{t-s}(\psi_{T_N-t}(u_k^x)+w) - \psi_{T_N-s}(u_k^x),
X_s\Big\rangle,
\end{align}
where the above expectation is finite for every $w\in\mathcal{I}^{k,x}$; recall
\eqref{eq:I_T}. This shows that $X$ is a time-inhomogeneous affine process
under $\P_k^x$, while \eqref{X-Pxk-characteristics-1} follows by substituting
$s=0$ in \eqref{X-Pkx-computation} and using the flow equations \eqref{flow}.
\end{proof}

\begin{remark}\label{r:forward-price-models}
The preservation of the affine property of the driving process under all 
forward measures is a stability property shared by all forward price models in  
which the process $1+\delta_x F_k^x = \frac{B(\cdot, T_{k-1}^x)}{B(\cdot, 
T_{k}^x)}$ is modeled as a deterministic exponential transformation of the 
driving process. This is related to the density process of the measure change 
between subsequent forward measures given exactly as $\frac{\ud \P_{k-1}^x}{\ud 
\P_{k}^x}\Big|_{\cF_t} =  \frac{B(0, T_{k}^x)}{B(0, T_{k-1}^x)} \frac{B(t, 
T_{k-1}^x)}{B(t, T_{k}^x)}= \frac{1+\delta_x F_k^x(t)}{1+\delta_x F_k^x(0)}$, 
see \eqref{Pkx-densities}, which is of the same exponential form, and 
guarantees that when performing a measure change the driving process remains in 
the same class. We refer to \citet{EberleinKluge06} for an example of a forward 
price model driven by a time-inhomogeneous L\'evy process under all forward 
measures. The models in the spirit of the LIBOR market model (LMM), where it is 
rather the forward rate $F_k^x$ which is modeled as an exponential, do not 
possess this property. The measure change in these models yields the stochastic 
terms $\frac{\delta_x F_k^x}{1+ \delta_x F_k^x} $ appearing in the 
characteristics of the driving process (more precisely, in the drift and in the 
compensator of the random measure of jumps) under any forward measure different 
from the terminal one, which destroys the analytical tractability of the model. 
The tractability is often re-established by freezing the value of these terms at 
their initial value $\frac{\delta_x F_k^x(0)}{1+ \delta_x F_k^x(0)} $  --- an 
approximation referred to as \textit{freezing the drift}. This approximation is 
widely known to be unreliable in many realistic settings; cf. 
\cite*{PapapantoleonSchoenmakersSkovmand10} and the references therein. 

On the other hand, in LMMs the positivity of the rate $F_k^x$ is ensured, which 
in general may not be the case in the forward price models. Due to their 
specific construction, the affine LIBOR models are able to reconcile both of 
these properties: the positivity of the rate $F_k^x$ and the structure 
preservation for the driving process under all forward measures. We refer the 
interested reader to a detailed discussion on this issue in Section 3 of 
\citet{KellerResselPapapantoleonTeichmann09}. Finally, it should be emphasized 
that in the current market situation the observed OIS rates have also negative 
values, but this situation can easily be included in the affine \lib models; 
cf. Section \ref{subs:nrps} below. 
\end{remark}

\begin{remark}[Single curve and deterministic spread]
The multiple curve \alm easily reduces to its single curve counterpart
(cf. \citealt{KellerResselPapapantoleonTeichmann09}) by setting $v_k^x=u_k^x$
for all $x\in\mathcal{X}$ and $k\in\mathcal{K}^x$. Another interesting question
is whether the spread can be deterministic or, similarly, whether the LIBOR rate
can be a deterministic transformation of the OIS rate. \par
Consider, for example, a 2-dimensional driving process $X=(X^1,X^2)$ where
$X^1$ is an arbitrary affine process and $X^2$ the constant process (i.e.
$X_t^2\equiv X_0^2)$. Then, by setting
\begin{align*}
u_{k-1}^x=(u_{1,k-1}^{x},0)
 \quad\text{ and }\quad
v_k^x=(u_{1,k-1}^{x},v_{2,k-1}^{x})
\end{align*}
where $u_{1,k-1}^{x},v_{2,k-1}^{x}>0$ we arrive at
\begin{align*}
 1+\delta_xL_k^x(t) = (1+\delta_xF_k^x(t)) \, \e^{v_{2,k-1}^{x}\cdot X_0^2}.
\end{align*}
Therefore, the LIBOR rate is a deterministic transformation of the OIS rate,
although the spread $S_k^x$ as defined in \eqref{spread} is not deterministic.
In that case, the multiplicative spread $R_k^x$ defined in \eqref{mlpl-spread}
is obviously deterministic.
\end{remark}

\subsection{A model with negative rates and positive spreads}\label{subs:nrps}

The multiple curve \alm produces positive rates and spreads, which is 
consistent with the typical market observations. However, in the current market 
environment negative rates have been observed, while the spreads still remain 
positive. Negative interest rates (as well as spreads, if needed) can be easily 
accommodated in this setup by considering, for example, affine processes on 
$\R^d$ instead of $\Rp^d$ or `shifted' positive affine processes where 
$\text{supp}(X)\in[a,\infty)^d$ with $a<0$.

In order to illustrate the flexibility of the affine LIBOR models, we provide 
below an explicit specification which allows for negative OIS rates, while still 
preserving the positive spreads. It is based on a particular choice of the 
driving process and suitable assumptions on the vectors $u_k^x$ and $v_k^x$. 
Recall from Remark \ref{r:choice-of-u-v} that if the driving process is 
multidimensional, we have a certain freedom in the choice of the parameters 
$u_k^x$ and $v_k^x$ when fitting the initial term structure, that we shall 
exploit here. 

Starting from the \alm in \eqref{eq:LIBOR-rate}, we have an expression for the 
OIS rates $F_k^x$ and we will derive an expression for the multiplicative 
spreads $R_k^x$ as defined in \eqref{mlpl-spread}. We choose the multiplicative 
spreads  as a more convenient quantity instead of the additive spreads $S_k^x$ 
in \eqref{spread}, but obviously the additive spreads can easily be recovered 
from the multiplicative spreads and the OIS rates, and vice versa, by combining 
\eqref{spread} and \eqref{mlpl-spread}. Moreover, the two spreads always have 
the same sign, i.e. $R_k^x \geq 0$ if and only if  $S_k^x \geq 0$. 

The model specification below allows in addition to ensure the monotonicity of 
the spreads with respect to the tenor length, which is also a feature typically 
observed in the markets.  More precisely, this means that for any two tenors 
$\mathcal{T}^{x_1}$ and $\mathcal{T}^{x_2}$ such that $\mathcal{T}^{x_2} \subset 
\mathcal{T}^{x_1}$, i.e. such that $\delta_{x_1}\leq \delta_{x_2}$,  the spreads 
have  the following property:  for all $k \in \mathcal{K}^{x_1}$ and $j \in 
\mathcal{K}^{x_2}$ such that $[T_{k-1}^{x_1}, T_{k}^{x_1}) \subset 
[T_{j-1}^{x_2}, T_{j}^{x_2})$ with $T_{k-1}^{x_1} = T_{j-1}^{x_2}$, we have  
$$
R_k^{x_1}(t) \leq R_j^{x_2}(t),
$$
for all $t \leq T_{k-1}^{x_1}$. That is, the spreads are lower for shorter tenor 
lengths. For example, a 3-month spread is lower than a 6-month spread for a 
6-month period starting at the same time as the 3-month period.

According to \eqref{eq:LIBOR-rate}, the OIS rate $F_k^x$, for every $x \in 
\mathcal{X}$ and every $k \in \mathcal{K}^x$, is provided by 
\begin{equation}\label{eq:neg-example-1}
1 +\delta_x  F_k^x(t) = \frac{M_t^{u_{k-1}^x}}{M_t^{u_{k}^x}},
\end{equation}
where we note that $u_k^x =u_l$, for $l$ such that $T_k^x = T_l$. This process 
is a $\P^x_k$-martingale by construction. The multiplicative spreads $R_k^x(t)$
take now the form
\begin{equation}\label{eq:neg-example-2}
1 +\delta_x  R_k^x(t) 
 = \frac{1 +\delta_x  L_k^x(t)}{1 +\delta_x  F_k^x(t)} 
 = \frac{M_t^{v_{k-1}^x}}{M_t^{u_{k-1}^x}},
\end{equation}
which is a $\P^x_{k-1}$-martingale by construction. 

Let us now present a possible choice of the driving process which allows to 
accommodate $F_k^x(t) \in \R$, while keeping $R_k^x(t) \in \R_{\geqslant0}$, as 
well as ensuring the monotonicity of the spreads with respect to the tenor 
length. We assume that the initial term structure of forward OIS rates $F_k^x(0) 
\in \R$ and of multiplicative spreads $R_k^x(0) \in \R_{\geqslant0}$ are given, 
for every fixed $x$ and all $k \in \mathcal{K}^x$.  Moreover, we assume that the 
initial spreads are monotone with respect to the tenor, i.e. for every two 
tenors $x_1$ and $x_2$ such that $\delta_{x_1}\leq \delta_{x_2}$, we have 
$R_k^{x_1}(0) \leq R_j^{x_2}(0)$, for all $k \in \mathcal{K}^{x_1}$ and $j \in 
\mathcal{K}^{x_2}$ with $T_{k-1}^{x_1} = T_{j-1}^{x_2}$.

In order to fix ideas, we shall consider only a 2-dimensional affine process 
$X=(X^1, X^2)$ on the state space $\R \times  \R_{\geqslant0}$ such that $X^1$ 
and $X^2$ are independent. The construction can easily be extended to 
$d$-dimensional affine processes on $\R^n \times  \R_{\geqslant0}^m$, with 
$n+m=d$, such that the first $n$ components are independent of the last $m$ 
components. The forward OIS rates will be driven by both components of the 
driving process $X$ and for the spreads we shall use only the second component 
$X^2$, which takes values in $ \R_{\geqslant0}$, to ensure the nonnegativity.  
This can be achieved by imposing appropriate assumptions on the parameters 
$v_k^x$. We split the construction in two steps.

\textit{Step 1.} Given the initial term structure of forward OIS rates $F_k^x(0) 
\in \R$, for every fixed $x$ and all $k \in \mathcal{K}^x$, we apply Proposition 
\ref{initial-fit-multiple-curve-OIS} and find a sequence $(u_k^x)\subset \R 
\times \R_{\geqslant0}$ such that the model \eqref{eq:neg-example-1} 
fits the initial term structure. Note that $u_k^x$, $k \in \mathcal K^x$, do 
not have to be ordered and $F_k^x(t) \in \R$, for any $t$. 

\textit{Step 2.} Next, given the initial term structure of multiplicative 
spreads $R_k^x(0) \in \R_{\geqslant0}$, for every fixed $x$ and all $k \in 
\mathcal{K}^x$, we calculate the initial \lib rates $L_k^x(0) $ using  
\eqref{mlpl-spread}. Applying Proposition \ref{initial-fit-multiple-curve} we 
can find a sequence  $(v_k^x) \subset \R \times  \R_{\geqslant0} $ such that 
for each $k \in \mathcal{K}^x$, $v_{k-1}^x = (v_{1, k-1}^{x}, v_{2, k-1}^{x})$ 
satisfies $v_{1, k-1}^{x} = u_{1, k-1}^{x}$ and the model 
\eqref{eq:neg-example-2} fits the initial term structure. Note that even though 
we fixed here the first component $v_{1, k-1}^{x}$ of each of the vectors 
$v_{k-1}^x$, Proposition \ref{initial-fit-multiple-curve} ensures that the 
initial term structure can be fitted using only the second components 
$v_{2, k-1}^{x} $. This yields 
\begin{multline}\label{eq:neg-example-3}
1 +\delta_x  R_k^x(t)   = \frac{M_t^{v_{k-1}^x}}{M_t^{u_{k-1}^x}}  \\
  = \frac{ \exp \left(  
\phi^1_{T_N-t}(u_{1, k-1}^{x}) \!+\! \phi^2_{T_N-t}(v_{2, k-1}^{x}) \!+\! 
\psi^1_{T_N-t}(u_{1, k-1}^{x}) X_t^1 \!+\! \psi^2_{T_N-t}(v_{2, k-1}^{x}) 
X_t^2\right) }{ \exp \left(  \phi^1_{T_N-t}(u_{1, k-1}^{x}) \!+\! 
\phi^2_{T_N-t}(u_{2, k-1}^{x}) \!+\! \psi^1_{T_N-t}(u_{1, k-1}^{x}) X_t^1 \!+\! 
\psi^2_{T_N-t}(u_{2, k-1}^{x}) X_t^2\right)} \\
  = \frac{\exp \left( 
\phi^2_{T_N-t}(v_{2, k-1}^{x}) + \psi^2_{T_N-t}(v_{2, k-1}^{x}) X_t^2\right) 
}{\exp \left( \phi^2_{T_N-t}(u_{2, k-1}^{x}) + \psi^2_{T_N-t}(u_{2, k-1}^{x}) 
X_t^2\right)},
\end{multline}
due to the independence of $X^1$ and $X^2$, see 
\citet[Prop.~4.7]{KellerRessel08}. Therefore, $R_k^x $ is driven only by $X^2$ 
and the fact that the initial values $R_k^x(0) \in \R_{\geqslant0}$ implies 
that 
$
v_{2, k-1}^{x} \geq u_{2, k-1}^{x}
$
for all $k$. Consequently, we have $R_k^x(t) \in \R_{\geqslant0}$, for all $t$, 
which follows immediately from \eqref{eq:neg-example-3}. 

Finally, it remains to show that the monotonicity of the initial spreads with 
respect to the tenor $R_k^{x_1}(0) \leq R_j^{x_2}(0)$ implies the monotonicity 
at all times $t$, $R_k^{x_1}(t) \leq R_j^{x_2}(t)$, for all $x_1, x_2$ and $k, 
j$ as above. First note that $T_{k-1}^{x_1} = T_{j-1}^{x_2}=T_l$ implies 
$u_{k-1}^{x_1} = u_{j-1}^{x_2}=u_l$ and consequently
\begin{equation}
\label{eq:neg-example-4}
M_t^{u_{k-1}^{x_1}} = M_t^{u_{j-1}^{x_2}}.
\end{equation}
Hence, $R_k^{x_1}(0) \leq R_{j}^{x_2}(0) $ implies that necessarily 
$M_0^{v_{k-1}^{x_1}} \leq  M_0^{v_{j-1}^{x_2}}  $  by \eqref{eq:neg-example-2}. 
This in turn yields  $v_{k-1}^{x_1} \leq v_{j-1}^{x_2}$, or more precisely 
$v_{2, k-1}^{x_1} \leq v_{2, j-1}^{x_2}$ since $v_{1, k-1}^{x_1} =  
u_{1,k-1}^{x_1} = u_{1, j-1}^{x_2} = v_{1, j-1}^{x_2}$.  As a consequence, 
$R_k^{x_1}(t) \leq 
R_{j}^{x_2}(t) $, for all $t$, since 
$$
1 +\delta_x  R_k^{x_1}(t) 
 =  \frac{M_t^{v_{k-1}^{x_1}}}{M_t^{u_{k-1}^{x_1}}} 
 \leq  \frac{M_t^{v_{j-1}^{x_2}}}{M_t^{u_{j-1}^{x_2}}} 
 = 1 +\delta_x  R_j^{x_2}(t),
$$
due to \eqref{eq:neg-example-4}.
\section{Connection to LIBOR market models} 
\label{alm-lmm}

In this section, we will clarify the relationship between the \alms and the 
`classical' \lib market models, cf. \citet*{SandmannSondermannMiltersen95} and 
\citet*{BraceGatarekMusiela97}, and also \citet{Mercurio_2010} for the 
extension of the LIBOR market models to multiple curves. This relationship has 
not yet been investigated even in the single-curve framework of 
\citet{KellerResselPapapantoleonTeichmann09}. More precisely, we will embed the 
multiple curve \alm  \eqref{eq:LIBOR-rate} in the general semimartingale LIBOR 
market model of \citet{Jamshidian97} and derive the corresponding dynamics of 
OIS and LIBOR rates. We shall concentrate on affine diffusion processes for the 
sake of simplicity, in order to expose the ideas without too many technical 
details. The generalization to affine processes with jumps is straightforward 
and left to the interested reader. 

An affine diffusion process on the state space $D=\Rp^d$ is the solution 
$X=X^\mathrm{x}$ of the SDE
\begin{equation}\label{X}
\ud X_t = (b+BX_t)\dt + \sigma(X_t) \ud W_t^N, \quad X_0=\mathrm{x},
\end{equation}
where $W^N$ is a $d$-dimensional $\P_N$-Brownian motion. The coefficients $b$, 
$B=(\beta_1,\dots,\beta_d)$ and $\sigma$ have to satisfy the admissibility 
conditions for affine diffusions on $\Rp^d$, see \citet[Ch.~10]{Filipovic09}.
That is, the drift vectors satisfy
\begin{align}
b\in\Rp^d, \quad 
\beta_{i(i)} \in \R 
  \quad \text{and} \quad 
\beta_{i(j)} \in \Rp \,\, 
  \quad \text{for all} \,\, 1\le i,j \le d, i\ne j,
\end{align}
where $\beta_{i(j)}$ denotes the $j$-th element of the column vector $\beta_i$. 
Moreover, the diffusion matrix $\sigma: D \to \R^{d\times d}$ satisfies
\begin{align}
\sigma(z)\sigma(z)^{\mathsf{T}} = \sum_{i=1}^d \alpha_iz_i,
 \quad \text {for all } z \in D,
\end{align}
where $\alpha_i$ are symmetric, positive semidefinite matrices for all 
$1\le i\le d$, such that
\begin{align}\label{AP-diff-matrix}
\alpha_{i(ii)} \in \Rp 
  \quad \text{and} \quad
\alpha_{i(jk)} = 0 \,\, 
  \quad \text{for all} \,\, 1\le i,j,k \le d, i\ne j,k.
\end{align}
Here, $\alpha_{i(jk)}$ denotes the $j,k$-th entry of the matrix $\alpha_i$. 
Therefore, the affine diffusion process $X$ is componentwise described by
\begin{align}\label{X-compw}
\ud X^i_t = (b+BX_t)^i\dt 
       + \sqrt{X^i_t}\, \sigma_i\, \ud W_t^N,
\end{align}
for all $i=1,\dots,d$, where $\sigma_i=\sqrt{\alpha_{i(ii)}}\cdot\e_i$ (with 
$\e_i$ the unit vector).

\subsection{OIS dynamics}

We start by computing the dynamics of OIS rates. As in the previous section, we 
consider a fixed $x\in\mathcal{X}$ and the associated tenor structure 
$\mathcal{T}^x$.

Using the structure of the $\P_N$-martingale $M^{u_k^x}$ in \eqref{eq:Mu}, we 
have that
\begin{align}\label{eq:dM}
\ud M_t^{u_k^x} = M_t^{u_k^x} \psi_{T_N-t}(u_k^x) \ud X_t + \dots\dt.
\end{align}
Hence, applying It\^o's product rule to \eqref{eq:LIBOR-rate} and using 
\eqref{eq:dM} yields that
\begin{align*}
\ud F_k^x(t)
 &= \frac{1}{\delta_x} \ud\frac{M_t^{u_{k-1}^x}}{M_t^{u_k^x}}
  = \frac{1}{\delta_x} \frac{M_{t}^{u_{k-1}^x}}{M_{t}^{u_k^x}}
    \left(\psi_{T_N-t}(u_{k-1}^x) \!-\! \psi_{T_N-t}(u_k^x)\right)
    \!\ud X_t  + \dots\dt\\
 &= \frac{1}{\delta_x} \left(1+\delta_xF_k^x(t)\right)  
    \left(\psi_{T_N-t}(u_{k-1}^x) - \psi_{T_N-t}(u_k^x)\right)
    \ud X_t  + \dots\dt.
\end{align*}
Therefore, the OIS rates satisfy the following SDE
\begin{align}
\frac{\ud F_k^x(t)}{F_k^x(t)}
 = \frac{1+\delta_xF_k^x(t)}{\delta_xF_k^x(t)}
   \left(\psi_{T_N-t}(u_{k-1}^x) - \psi_{T_N-t}(u_k^x)\right) \ud X_t
 + \dots\dt
\end{align}
for all $k=2,...,N^x$. Now, using the dynamics of the affine process $X$ from 
\eqref{X-compw} we arrive at
\begin{align}
\frac{\ud F_k^x(t)}{F_k^x(t)} 
 &= \frac{1+\delta_xF_k^x(t)}{\delta_xF_k^x(t)}
   \sum_{i=1}^{d} \left(\psi_{T_N-t}^i(u_{k-1}^x) - \psi_{T_N-t}^i(u_k^x)\right)
    \ud X_t^i + \dots\dt \nonumber\\
 &\!\!\!\!\!\!\!\!\! = \frac{1+\delta_xF_k^x(t)}{\delta_xF_k^x(t)}
   \sum_{i=1}^{d} \left(\psi_{T_N-t}^i(u_{k-1}^x) - \psi_{T_N-t}^i(u_k^x)\right)
    \sqrt{X_t^i} \, \sigma_i \ud W_t^N + \dots\dt \nonumber\\
 &=: \Gamma_{x,k}(t) \, \ud W_t^N + \dots\dt,
\end{align}
where we define the volatility structure
\begin{align}\label{gxk}
\Gamma_{x,k}(t)
 &= \frac{1+\delta_xF_k^x(t)}{\delta_xF_k^x(t)} \sum_{i=1}^{d} 
    \left( \psi_{T_N-t}^i(u_{k-1}^x) -\psi_{T_N-t}^i(u_k^x) \right)
    \sqrt{X_t^i} \, \sigma_i\in\Rp^d.
\end{align}

On the other hand, we know from the general theory of discretely compounded 
forward rates (cf. \citealt{Jamshidian97}) that the OIS rate should satisfy the 
following SDE under the terminal measure $\P_N$
\begin{align}
\frac{\ud F_k^x(t)}{F_k^x(t)} 
 &= - \sum_{l=k+1}^{N^x} \frac{\delta_xF_l^x(t)}{1+\delta_xF_l^x(t)}
      \scal{\Gamma_{x,l}(t)}{\Gamma_{x,k}(t)} \,\dt
    + \Gamma_{x,k}(t) \,\ud W_t^N,
\end{align}
for the volatility structure $\Gamma_{x,k}$ given in \eqref{gxk}. Therefore, we 
get immediately that the $\P_k^x$-Brownian motion $W^{x,k}$ is related to the 
terminal Brownian motion $W^N$ via the equality
\begin{align}\label{Pkx-BM}
W^{x,k} 
 &:= W^N - \sum_{l=k+1}^{N^x} \int_0^\cdot 
     \frac{\delta_xF_l^x(t)}{1+\delta_xF_l^x(t)}\Gamma_{x,l}(t)\,\dt \nonumber\\
 &= W^N - \sum_{l=k+1}^{N^x} \sum_{i=1}^{d} \int_0^\cdot
    \left( \psi_{T_N-t}^i(u_{l-1}^x) -\psi_{T_N-t}^i(u_l^x) \right)
    \sqrt{X_t^i} \, \sigma_i \,\dt.
\end{align}
Moreover, the dynamics of $X$ under $\P_k^x$ take the form
\begin{align}\label{Pkx-X}
\ud X_t^i 
 &= \left(b+BX_t\right)^i\dt + \sqrt{X_t^i} \, \sigma_i
    \ud W_t^{x,k} \nonumber\\
 &\quad+ \sigma_i \sqrt{X_t^i} \sum_{l=k+1}^{N^x} \sum_{j=1}^{d}
    \left(\psi_{T_N-t}^{j} \left(u_{l-1}^x\right)
         -\psi_{T_N-t}^{j} \left(u_l^x\right)\right) 
      \sqrt{X_t^{j}}\sigma_j\dt \nonumber\\
 &= \left(b_i+\left(BX_t\right)^i + \sum_{l=k+1}^{N^x}
     \left(\psi_{T_N-t}^i(u_{l-1}^x)-\psi_{T_N-t}^i(u_l^x)\right)
      X_t^i \, |\sigma_i|^2 \right) \dt \nonumber\\
 &\quad+ \sqrt{X_t^i} \, \sigma_i \ud W_t^{x,k},
\end{align}
for all $i=1,\dots,d$. The last equation provides an alternative proof to 
Proposition \ref{X-Pkx-characteristics} in the setting of affine diffusions, 
since it shows explicitly that $X$ is a time-inhomogeneous affine diffusion 
process under $\P_k^x$. One should also note from \eqref{Pkx-BM}, that the 
difference between the terminal and the forward Brownian motion does not depend 
on other forward rates as in `classical' LIBOR market models. As mentioned in 
Remark \ref{r:forward-price-models}, the same property is shared by forward 
price models.

Thus, we arrive at the following $\P_k^x$-dynamics for the OIS rates
\begin{align}\label{OIS-Pkx-SDE}
\frac{\ud F_k^x(t)}{F_k^x(t)} 
 &= \Gamma_{x,k}(t) \, \ud W_t^{x,k}
\end{align}
with the volatility structure $\Gamma_{x,k}$ provided by \eqref{gxk}. The 
structure of $\Gamma_{x,k}$ shows that there is a built-in shift in the model, 
whereas the volatility structure is determined by $\psi$ and $\sigma$.

\subsection{LIBOR dynamics}

Next, we derive the dynamics of the LIBOR rates associated to the same tenor. 
Using \eqref{eq:LIBOR-rate}, \eqref{eq:Mv} and repeating the same steps as 
above, we obtain the following
\begin{align*}
\frac{\ud L_k^x(t)}{L_k^x(t)} 
 &= \frac{1}{\delta_xL_k^x(t)} \ud \frac{M_t^{v_{k-1}^x}}{M_t^{u_k^x}} \\
 &= \frac{1}{\delta_xL_k^x(t)} \frac{M_t^{v_{k-1}^x}}{M_t^{u_k^x}}
    \left(\psi_{T_N-t}(v_{k-1}^x)-\psi_{T_N-t}(u_k^x)\right) \ud X_t 
  + \dots\dt\\
 &= \frac{1+\delta_xL_k^x(t)}{\delta_xL_k^x(t)} \sum_{i=1}^{d}
    \left(\psi_{T_N-t}^i(v_{k-1}^x)\!-\!\psi_{T_N-t}^i(u_k^x)\right) 
    \!\sqrt{X_t^i}
    \, \sigma_i \ud W_t^N + \dots\dt,
\end{align*}
for all $k=2,...,N^{x}.$ Similarly to \eqref{gxk} we introduce the volatility 
structure
\begin{equation}\label{exk}
\Lambda_{x,k}(t) := \frac{1+\delta_xL_k^x(t)}{\delta_xL_k^x(t)}
  \sum_{i=1}^d \left(\psi_{T_N-t}^i(v_{k-1}^x)\!-\!\psi_{T_N-t}^i(u_k^x)\right)
  \!\sqrt{X_t^i} \, \sigma_i \in\Rp^{d},
\end{equation}
and then obtain for $L_k^x$ the following $\P_k^x$-dynamics
\begin{align}\label{LIBOR-dynamics}
\frac{\ud L_k^x(t)}{L_k^x(t)} = \Lambda_{x,k}(t) \, \ud W_t^{x,k},
\end{align}
where $W^{x,k}$ is the $\P_k^x$-Brownian motion given by \eqref{Pkx-BM}, while 
the dynamics of $X$ are provided by \eqref{Pkx-X}.

\subsection{Spread dynamics}

Using that $S_k^x=L_k^x-F_k^x$, the dynamics of LIBOR and OIS rates under the 
forward measure $\P_k^x$ in \eqref{OIS-Pkx-SDE} and \eqref{LIBOR-dynamics}, as 
well as the structure of the volatilities in \eqref{gxk} and \eqref{exk}, after 
some straightforward calculations we arrive at
\begin{align*}
\ud S_k^x(t)
 &= \left\{ S_k^x(t) \Upsilon_t (v_{k-1}^x,u_k^x)
  + \frac{1+\delta_x F_k^x(t)}{\delta_x}      
    \Upsilon_t (v_{k-1}^x,u_{k-1}^x) \right\} \ud W_t^{x,k},
\end{align*}
where
\begin{align}
\Upsilon_t(w,y)
 := \sum_{i=1}^d \left( 
    \psi_{T_N-t}^i(w)\!-\!\psi_{T_N-t}^i(y)\right) 
    \!\sqrt{X_t^i}\,\sigma_i.
\end{align}

\subsection{Instantaneous correlations}
\label{inst_corr}

The derivation of the SDEs that OIS and LIBOR rates satisfy allows to provide 
quickly formulas for various quantities of interest, such as the instantaneous 
correlations between OIS and LIBOR rates or LIBOR rates with different 
maturities or tenors. We have, for example, that the instantaneous correlation 
between the LIBOR rates maturing at $T_k^x$ and $T_l^x$ is heuristically 
described by
\begin{align*}
\text{Corr}_t\big[L_k^x,L_l^x\big]
 = \frac{\frac{\ud L_k^x(t)}{L_k^x(t)} \cdot \frac{\ud L_l^x(t)}{L_l^x(t)}}
        {\sqrt{\frac{\ud L_k^x(t)}{L_k^x(t)} \cdot
               \frac{\ud L_k^x(t)}{L_k^x(t)}}
	 \sqrt{\frac{\ud L_l^x(t)}{L_l^x(t)} \cdot
	       \frac{\ud L_l^x(t)}{L_l^x(t)}}}
\end{align*}
therefore we get that
\begin{multline*}
\text{Corr}_t\big[L_k^x,L_l^x\big]
 \stackrel{\eqref{LIBOR-dynamics}}{=} 
   \frac{\scal{\Lambda_{x,k}}{\Lambda_{x,l}}}
        {|\Lambda_{x,k}||\Lambda_{x,l}|} \\
 = \frac{\sum_{i=1}^{d} \left( \psi_{T_N-t}^i\left(v_{k-1}^{x}\right)
     - \psi_{T_N-t}^i\left(u_{k}^{x}\right)  \right)  
     \left( \psi_{T_N-t}^i\left(v_{l-1}^x\right) 
     - \psi_{T_N-t}^i\left(u_{l}^{x}\right)  \right)  X^i |\sigma_i|^2}
    {\sqrt{\sum_{i=1}^{d} \left(\psi_{T_N-t}^i\left(v_{k-1}^{x}\right)
     - \psi_{T_N-t}^i\left(u_{k}^{x}\right) \right)^{2} X^i |\sigma_i|^2}}
  \\ \times \frac{1}{ \sqrt{\sum_{i=1}^{d} \left(
       \psi_{T_N-t}^i\left(v_{l-1}^x\right) - \psi_{T_N-t}^i\left(u_l^x\right) 
        \right)^{2} X^i |\sigma_i|^2}}.
\end{multline*}
Similar expressions can be derived for other instantaneous correlations, e.g.
\begin{align*}
\text{Corr}_t\big[F_k^x,L_k^x\big] 
 \quad\text{ or }\quad
\text{Corr}_t\big[L_k^{x_1},L_k^{x_2}\big].
\end{align*}

Instantaneous correlations are important for describing the (instantaneous) 
interdependencies between different \lib rates. In the \lib market model for 
instance, the rank of the instantaneous correlation matrix determines the 
number of factors (e.g. Brownian motions) that is needed to drive the model. 
Explicit expressions for terminal correlations between LIBOR rates are provided 
in Appendix \ref{app-corr}.
\section{Valuation of swaps and caps}
\label{caps}

\subsection{Interest rate and basis swaps}
\label{s:swaps}

We start by presenting a fixed-for-floating payer interest rate swap on a 
notional amount normalized to $1$, where fixed payments are exchanged for 
floating payments linked to the LIBOR rate. The LIBOR rate is set in advance and 
the payments are made in arrears, while we assume for simplicity that the timing 
and frequency of the payments of the floating leg coincides with those of the 
fixed leg. The swap is initiated at time $T_p^x  \geq 0$, where 
$x\in\mathcal{X}$ and $p \in \mathcal{K}^x$. The collection of payment dates is 
denoted by $\mathcal{T}^x_{pq}:=\{T^x_{p+1} < \cdots < T^x_{q}\}$, and the fixed 
rate is denoted by $K$. Then, the time-$t$ value of the swap, for $t\leq 
T^{x}_p$, is given by
\begin{align}\label{eq:swap-value}
\nonumber 
\mathbb{S}_t(K, \mathcal{T}^x_{pq}) 
 &= \sum_{k=p+1}^{q} \delta_x B(t, T^x_k) \, 
    \E^{x}_k \big[L(T^x_{k-1}, T^x_k) - K | \mathcal{F}_t\big] \\
 &= \delta_x \sum_{k=p+1}^{q}  B(t, T^x_k) \left(  L_k^x(t) - K \right).
\end{align}
Thus, the \textit{fair swap rate} $K_t(\mathcal{T}^x_{pq})$  is provided by
\begin{align}\label{eq:swap-rate}
K_t(\mathcal{T}^x_{pq}) 
 &= \frac{\sum_{k=p+1}^{q} B(t, T^x_k) L_k^x(t)}{\sum_{k=p+1}^{q} B(t, T^x_k)}.
\end{align}

Basis swaps are new products in interest rate markets, whose value reflects the 
discrepancy between the LIBOR rates of \textit{different} tenors. A basis swap 
is a swap where two streams of floating payments linked to the LIBOR rates of 
different tenors are exchanged. For example, in a 3m--6m basis swap, a 3m-LIBOR 
is paid (received) quarterly and a 6m-LIBOR is received (paid) semiannually. We 
assume in the sequel that both rates are set in advance and paid in arrears; of 
course, other conventions regarding the payments on the two legs of a basis 
swap also exist. A more detailed account on basis swaps can be found in 
\citet[Section 5.2]{Mercurio_2010a} or in 
\citet[Section~2.4~and~Appendix~F]{Filipovic_Trolle_2011}. Note that in the 
pre-crisis setup the value of such a product would have been zero at any time 
point, due to the no-arbitrage relation between the LIBOR rates of different 
tenors; see e.g. \citet*{Crepey_Grbac_Nguyen_2011}.

Let us consider a basis swap associated with two tenor structures denoted by 
$\mathcal{T}^{x_1}_{pq}:=\{T^{x_1}_{p_1}<\ldots<T^{x_1}_{q_1}\}$ and 
$\mathcal{T}^{x_2}_{pq}:=\{T^{x_2}_{p_2}<\ldots<T^{x_2}_{q_2}\}$, where 
$T^{x_1}_{p_1}=T^{x_2}_{p_2}\geq0$, $T^{x_1}_{q_1}=T^{x_2}_{q_2}$ and
$\mathcal{T}^{x_2}_{pq}\subset\mathcal{T}^{x_1}_{pq}$. The notional amount is 
again assumed to be $1$ and the swap is initiated at time $T^{x_1}_{p_1}$, 
while the first payments are due at times $T^{x_1}_{p_1+1}$ and 
$T^{x_2}_{p_2+1}$ respectively. The basis swap spread is a fixed rate $S$
which is added to the payments on the shorter tenor length. More precisely, for 
the $x_1$-tenor, the floating interest rate $L(T^{x_1}_{i-1},T^{x_1}_i)$ at 
tenor date $T^{x_1}_i$ is replaced by $L(T^{x_1}_{i-1},T^{x_1}_i)+S$, for every 
$i\in\{p_1+1,\ldots,q_1\}$. The time-$t$ value of such an agreement is given,
for $0\le t\le T^{x_1}_{p_1}=T^{x_2}_{p_2}$, by
\begin{multline}\label{eq:basis-swap-aux}
\mathbb{BS}_t(S,\mathcal{T}^{x_1}_{pq},\mathcal{T}^{x_2}_{pq}) 
 = \sum_{i=p_2+1}^{q_2} \delta_{x_2} B(t,T^{x_2}_i) \,
    \E^{x_2}_{i} \left[L(T^{x_2}_{i-1},T^{x_2}_i) | \cF_t\right]  \\
 \qquad\qquad - \sum_{i=p_1+1}^{q_1} \delta_{x_1} B(t,T^{x_1}_i) \,
    \E^{x_1}_i \left[L(T^{x_1}_{i-1}, T^{x_1}_i) + S | \cF_t \right] \\
 = \sum_{i=p_2+1}^{q_2} \delta_{x_2} B(t,T^{x_2}_i)  L^{x_2}_i(t) \,
  - \sum_{i=p_1+1}^{q_1} \delta_{x_1} B(t,T^{x_1}_i) \big(L^{x_1}_i(t)+S\big).
\end{multline}
We also want to compute the \textit{fair basis swap spread} 
$S_t(\mathcal{T}^{x_1}_{pq},\mathcal{T}^{x_2}_{pq})$. This is the 
spread that makes the value of the basis swap equal zero at time $t$, i.e. it 
is obtained by solving 
$\mathbb{BS}_t(S,\mathcal{T}^{x_1}_{pq},\mathcal{T}^{x_2}_{pq})=0$. We get that
\begin{equation}\label{eq:basis-swap-spread}
S_t(\mathcal{T}^{x_1}_{pq},\mathcal{T}^{x_2}_{pq})
 = \frac{\sum_{i=p_2+1}^{q_2} \delta_{x_2}  B(t,T^{x_2}_i)  L^{x_2}_i(t) 
	- \sum_{i=p_1+1}^{q_1}  \delta_{x_1}  B(t,T^{x_1}_i) L^{x_1}_i(t)}
	{\sum_{i=p_1+1}^{q_1} \delta_{x_1} B(t,T^{x_1}_i)}.
\end{equation}
The formulas for the fair swap rate and basis spread can be used to bootstrap 
the initial values of LIBOR rates from market data, see 
\citet[\S2.4]{Mercurio_2010a}.

\subsection{Caps}
\label{subsection:Caplets}

The valuation of caplets, and thus caps, in the multiple curve \alm is an easy 
task, which has complexity equal to the complexity of the valuation of caplets 
in the single-curve \alm; compare with Proposition 7.1 in 
\citet{KellerResselPapapantoleonTeichmann09}. There are two reasons for this: 
on the one hand, the LIBOR rate is modeled directly, see \eqref{eq:LIBOR-rate}, 
as opposed to e.g. \citet{Mercurio_2010} where the LIBOR rate is modeled 
implicitly as the sum of the OIS rate and the spread. In our approach, the 
valuation of caplets remains a one-dimensional problem, while in the latter 
approach it becomes a `basket' option on the OIS rate and the spread. On the 
other hand, the driving process remains affine under any forward measure, cf. 
Proposition \ref{X-Pkx-characteristics}, which allows the application of Fourier 
methods for option pricing. In the sequel, we will derive semi-explicit pricing 
formulas for any multiple curve \alm. Let us point out that we do not need to 
`freeze the drift' as is customary in \lib market models with jumps (see Remark 
\ref{r:forward-price-models}).

\begin{proposition}
Consider an $x$-tenor caplet with strike $K$ that pays out 
$\delta_x(L(T_{k-1}^x,T_k^x)-K)^+$ at time $T_k^x$. The time-0 price is provided 
by
\begin{align}\label{MC-ALM-caplets}
\mathbb{C}_0(K,T_k^x)
 &= \frac{B(0,T_k^x)}{2\pi} \int_\R K_x^{1-R+iw}
     \frac{\Theta_{\mathcal{W}^x_{k-1}}(R-iw)}{(R-iw)(R-1-iw)} \dw,
\end{align}
for $R\in(1,\infty)\cap\widetilde{\mathcal{I}}^{k,x}$, assuming that 
$(1,\infty)\cap\widetilde{\mathcal{I}}^{k,x}\neq\emptyset$, where 
$K_x=1+\delta_xK$, $\Theta_{\mathcal{W}^x_{k-1}}$ is given by \eqref{caplet-2}, 
while the set $\widetilde{\mathcal{I}}^{k,x}$ is defined as 
\begin{align*}
\widetilde{\mathcal{I}}^{k,x} 
 = \set{z\in\R: (1-z)\psi_{T_N-T_{k-1}^x}(u_k^x)
        + z\psi_{T_N-T_{k-1}^x}(v_{k-1}^x) \in \mathcal{I}_T}.
\end{align*}
\end{proposition}

\begin{proof}
Using \eqref{LIBOR-defin}  and \eqref{eq:LIBOR-rate} the time-0 price of 
the caplet equals
\begin{align*}
\mathbb{C}_0(K,T_k^x)
 &= \delta_x \, B(0,T_k^x) \, \E_k^x\big[ (L(T_{k-1}^x, T_k^x) - K)^+ \big] \\
 &= \delta_x \, B(0,T_k^x) \, \E_k^x\big[ (L_k^x(T_{k-1}^x) - K)^+ \big] \\
 &= B(0,T_k^x) \, \E_k^x\Big[ \Big(
    M_{T_{k-1}^x}^{v_{k-1}^x} / M_{T_{k-1}^x}^{u_k^x} - K_x \Big)^+ \Big] \\
 &= B(0,T_k^x) \, 
    \E_k^x\Big[ \Big(\e^{\mathcal{W}_{k-1}^x} - K_x\Big)^+ \Big],
\end{align*}
where
\begin{align}\label{caplet-1}
\mathcal{W}^x_{k-1}
 &= \log\left( M_{T_{k-1}^x}^{v_{k-1}^x} / M_{T_{k-1}^x}^{u_k^x} \right)\notag\\
 &= \phi_{T_N-T_{k-1}^x}(v_{k-1}^x) - \phi_{T_N-T_{k-1}^x}(u_k^x) \notag\\
 &\quad + \big\langle \psi_{T_N-T_{k-1}^x}(v_{k-1}^x)
        - \psi_{T_N-T_{k-1}^x}(u_k^x),X_{T^x_{k-1}} \big\rangle \notag\\
 &=: A + \scal{B}{X_{T^x_{k-1}}}.
\end{align}
Now, using \citet*[Thm 2.2, Ex. 5.1]{EberleinGlauPapapantoleon08}, we arrive 
directly at \eqref{MC-ALM-caplets}, where $\Theta_{\mathcal{W}^x_{k-1}}$ 
denotes the $\P_k^x$-moment generating function of the random variable 
$\mathcal{W}_{k-1}^x$, i.e. for $z\in\widetilde{\mathcal{I}}^{k,x}$,
\begin{align}\label{caplet-2}
\Theta_{\mathcal{W}^x_{k-1}}(z)
 &= \E_k^x\big[ \e^{z\mathcal{W}_{k-1}^x} \big] 
  = \E_k^x\big[ \exp\big( z(A + \scal{B}{X_{T^x_{k-1}}}) \big) \big] \notag\\
 &= \exp\Big( zA + \phi^{k,x}_{T^x_{k-1}}(zB)
             + \big\langle\psi^{k,x}_{T^x_{k-1}}(zB),X_0\big\rangle \Big).
\end{align}
The last equality follows from Proposition \ref{X-Pkx-characteristics}, noting 
that $z \in \widetilde{\mathcal{I}}^{k,x}$ implies $z B \in \mathcal{I}^{k,x}$.
\end{proof}
\section{Valuation of swaptions and basis swaptions}
\label{val:swaptions}

This section is devoted to the pricing of options on interest rate and basis 
swaps, in other words, to the pricing of swaptions and basis swaptions. In the 
first part, we provide general expressions for the valuation of swaptions and 
basis swaptions making use of the structure of multiple curve \alms. In the 
following two parts, we derive efficient and accurate approximations for the 
pricing of swaptions and basis swaptions by further utilizing the model 
properties, namely the preservation of the affine structure under any forward 
measure, and applying the linear boundary approximation developed by 
\citet{SingletonUmantsev02}. Similarly to the pricing of caplets, also here  we 
do not have to `freeze the drift', while in special cases we can even derive 
closed or semi-closed form solutions (cf. 
\citealt[\S8]{KellerResselPapapantoleonTeichmann09}).

Let us consider first a payer swaption with strike rate $K$ and exercise date 
$T^x_p$ on a fixed-for-floating interest rate swap starting at $T^x_p$ and 
maturing at $T^x_q$; this was defined in Section \ref{s:swaps}. A swaption can 
be regarded as a sequence of fixed payments 
$\delta_x(K_{T^x_p}(\mathcal{T}^x_{pq})-K)^+$ that are received at the payment 
dates $T^x_{p+1},\ldots,T^x_q$; see \citet[Section~13.1.2, 
p.~524]{MusielaRutkowski05}. Here $K_{T^x_p}(\mathcal{T}^x_{pq})$ is the swap 
rate of the underlying swap at time $T^x_p$, cf. \eqref{eq:swap-rate}. Note that 
the classical transformation of a payer (resp. receiver) swaption into a put 
(resp. call) option on a coupon bond is not valid in the multiple curve setup, 
since \lib rates cannot be expressed in terms of zero coupon bonds; see Remark 
\ref{LIBOR-neq-ZCB}.

The value of the swaption at time $t \leq T^x_{p}$ is provided by
\begin{multline*}
\mathbb{S}^+_t(K,\mathcal{T}^x_{pq})
 = B(t,T^x_p) \sum_{i=p+1}^{q} \delta_x \, \E^{x}_p \left[ B(T^x_p,T^x_i)
    \left( K_{T^x_p}(\mathcal{T}^x_{pq}) - K \right)^+ \Big| \cF_t \right] \\
 = B(t,T^x_p) \, \E^{x}_p \left[ \left( \sum_{i=p+1}^{q} \delta_x
    L^x_i(T^x_p) B(T^x_p,T^x_i)
  - \sum_{i=p+1}^{q} \delta_x K B(T^x_p,T^x_i) \right)^+  \Big| \cF_t  \right]
\end{multline*}
since the swap rate $K_{T^x_p}(\mathcal{T}^x_{pq})$ is given by 
\eqref{eq:swap-rate} for $t=T_p^x$. Using \eqref{OIS-defin}, 
\eqref{eq:LIBOR-rate} and a telescoping product, we get that
\[
 B(T^x_p,T^x_i) 
 = \frac{B(T^x_p,T^x_i)}{B(T^x_p,T^x_{i-1})} 
   \frac{B(T^x_p,T^x_{i-1})}{B(T^x_p,T^x_{i-2})}  \cdots  
   \frac{B(T^x_p,T^x_{p+1})}{B(T^x_p,T^x_{p})} 
 = \frac{M_{T_p^x}^{u_i^x}}{M_{T_p^x}^{u_{p}^x}}.
\]
Together with \eqref{eq:LIBOR-rate} for   $L^x_i(T^x_p)$, this yields
\begin{align}\label{val-swaption-general} 
\mathbb{S}^+_t(K,\mathcal{T}^x_{pq})
 &= B(t,T^x_p) \, \E^{x}_p \left[  \left( \sum_{i=p+1}^{q}
     \frac{M_{T_p^x}^{v_{i-1}^x}}{M_{T_p^x}^{u_p^x}} - \sum_{i=p+1}^{q} K_x
     \frac{M_{T_p^x}^{u_i^x}}{M_{T_p^x}^{u_p^x}} \right)^+  \Big| \cF_t  \right] \nonumber \\
 &= B(t, T_N) \, \E_N \left[ \left( \sum_{i=p+1}^{q} M_{T_p^x}^{v_{i-1}^x}
    - \sum_{i=p+1}^{q} K_x M_{T_p^x}^{u_i^x} \right)^+  \Big| \cF_t  \right],
\end{align}
where $K_x:=1+\delta_xK$ and the second equality follows from the measure change from $\P_p^x$ to  
$\P_N$ as given in \eqref{Pkx-densities}.

Next, we move on to the pricing of basis swaptions. A basis swaption is an
option to enter a basis swap with spread $S$. We consider a basis swap as
defined in Section \ref{s:swaps}, which starts at $T^{x_1}_{p_1}=T^{x_2}_{p_2}$
and ends at $T^{x_1}_{q_1}=T^{x_2}_{q_2}$, while we assume that the exercise
date is $T^{x_1}_{p_1}$. The payoff of a basis swap at time $T^{x_1}_{p_1}$ is
given by \eqref{eq:basis-swap-aux} for $t=T^{x_1}_{p_1}$. Therefore, the price
of a basis swaption at time $t \leq T^x_{p}$ is provided by
\begin{align*}
\mathbb{BS}^+_t(S,\mathcal{T}^{x_1}_{pq},\mathcal{T}^{x_2}_{pq})
 &= B(t,T^{x_1}_{p_1}) \, \E^{x_1}_{p_1} \left[ \left( \sum_{i=p_2+1}^{q_2}
    \delta_{x_2} L^{x_2}_i(T^{x_2}_{p_2}) B(T^{x_2}_{p_2},T^{x_2}_i)
     \right. \right. \\ &\qquad\qquad\quad \left. \left.
  - \sum_{i=p_1+1}^{q_1}  \delta_{x_1} \big(L^{x_1}_i(T^{x_1}_{p_1}) + S\big)
     B(T^{x_1}_{p_1},T^{x_1}_i) \right)^+   \Big| \cF_t \right].
\end{align*}
Along the lines of the derivation for swaptions and using 
$M_{T_{p_2}^{x_2}}^{u_{p_2}^{x_2}}=M_{T_{p_1}^{x_1}}^{u_{p_1}^{x_1}}$ 
(cf. \eqref{eq:u}), we arrive at
\begin{multline}\label{val-bswaption-general}
\mathbb{BS}^+_0(S,\mathcal{T}^{x_1}_{pq},\mathcal{T}^{x_2}_{pq}) = \\
 = B(t,T^{x_1}_{p_1}) \, \E^{x_1}_{p_1} \left[ \left( \sum_{i=p_2+1}^{q_2}
    \left( M_{T_{p_2}^{x_2}}^{v_{i-1}^{x_2}} / M_{T_{p_2}^{x_2}}^{u_{p_2}^{x_2}}
   - M_{T_{p_2}^{x_2}}^{u_{i}^{x_2}} / M_{T_{p_2}^{x_2}}^{u_{p_2}^{x_2}} \right)
 \right.\right. \\ \left.\left.
   - \sum_{i=p_1+1}^{q_1} \left(
     M_{T_{p_1}^{x_1}}^{v_{i-1}^{x_1}} / M_{T_{p_1}^{x_1}}^{u_{p_1}^{x_1}}
   - S_{x_1} M_{T_{p_1}^{x_1}}^{u_i^{x_1}} / M_{T_{p_1}^{x_1}}^{u_{p_1}^{x_1}}
  \right)\right)^+  \Big| \cF_t  \right] \\
 = B(t, T_N) \, \E_N \! \left[ \left( \sum_{i=p_2+1}^{q_2}
   \! \left( \! M_{T_{p_2}^{x_2}}^{v_{i-1}^{x_2}}
    \!-\!  M_{T_{p_2}^{x_2}}^{u_{i}^{x_2}} \! \right)
   -\! \sum_{i=p_1+1}^{q_1} \! \left(\! M_{T_{p_1}^{x_1}}^{v_{i-1}^{x_1}}
    \!-\! S_{x_1} M_{T_{p_1}^{x_1}}^{u_i^{x_1}} \! \right) \! \right)^+  \Big| \cF_t  \right],
\end{multline}
where $S_{x_1}:=1-\delta_{x_1}S$.

\subsection{Approximation formula for swaptions}
\label{section:Linear_boundary_approximation}

We will now derive an efficient approximation formula for the pricing of
swaptions. The main ingredients in this formula are the affine property of the
driving process under forward measures and the linearization of the exercise
boundary. Numerical results for this approximation will be reported in Section
\ref{sec:swaption-ntest}.

We start by presenting some technical tools and assumptions that will be used in
the sequel. We define the probability measures $\overline{\P}_k^x$, for every
$k\in\mathcal{K}^x$, via the Radon--Nikodym density
\begin{align}\label{Pkx-overline-densities}
\frac{\ud \overline{\P}_k^x}{\ud \P_N}\Big|_{\cF_t}
 = \frac{M^{v_k^x}_t}{M^{v_k^x}_0}.
\end{align}
The process $X$ is obviously a time-inhomogeneous affine process under every
$\overline{\P}_k^x$. More precisely, we have the following result which follows
directly from Proposition \ref{X-Pkx-characteristics}.

\begin{corollary}\label{X-Pkx-overline-characteristics}
The process $X$ is a time-inhomogeneous affine process under the measure
$\overline{\P}_k^x$, for every $x\in\mathcal{X},k\in\mathcal{K}^x$, with
\begin{align}\label{bPkx-mgf}
\bE_{k}^x \big[ \e^{\scal{w}{X_t}} \big]
 &= \exp\left( \bphi_t^{k,x}(w) + \scal{\bpsi_t^{k,x}(w)}{X_0} \right),
\end{align}
where
\begin{subequations}\label{bPkx-phipsi}
\begin{align}
\bphi_t^{k,x}(w) &:= \phi_t\big(\psi_{T_N-t}(v_k^x)+w\big)
                   - \phi_t\big(\psi_{T_N-t}(v_k^x)\big),\\
\bpsi_t^{k,x}(w) &:= \psi_t\big(\psi_{T_N-t}(v_k^x)+w\big)
                   - \psi_t\big(\psi_{T_N-t}(v_k^x)\big),
\end{align}
\end{subequations}
for every $w \in \overline{\mathcal{I}}^{k,x}$ with
\begin{align}\label{eq:II-kx}
\overline{\mathcal{I}}^{k,x}
 := \set{w\in\R^d: \psi_{T_N-t}(v_k^x) + w \in \mathcal{I}_T}.
\end{align}
\end{corollary}

The price of a swaption is provided by \eqref{val-swaption-general}, while for 
simplicity we shall consider the price at time $t=0$ in the sequel. We can 
rewrite \eqref{val-swaption-general} as follows
\begin{align}\label{eq:swaption-price-with-f}
\nonumber \mathbb{S}^+_0(K,\mathcal{T}^x_{pq}) & 
  = B(0, T_N) \, \E_N \! \left[ \left( \sum_{i=p+1}^{q}
   M_{T_p^x}^{v_{i-1}^x} - \sum_{i=p+1}^{q} K_x M_{T_p^x}^{u_i^x}
   \right) \! \indik_{\{f(X_{T_p^x}) \geq 0\}}\right] \\
 & = B(0, T_N) \left( \sum_{i=p+1}^{q}  \E_N \left[
   M_{T_p^x}^{v_{i-1}^x} \indik_{\{f(X_{T_p^x}) \geq 0\}} \right] \right. \\
   \nonumber  & \qquad \qquad \left. - K_x \sum_{i=p+1}^{q} \E_N \left[  
    M_{T_p^x}^{u_i^x} \indik_{\{f(X_{T_p^x}) \geq 0\}} \right] \right),
\end{align}
where, recalling \eqref{eq:Mu} and \eqref{eq:Mv}, we define the function 
$f:\Rp^d\to\R$ by
\begin{align}\label{eq:f_swaption}
f(y) &= \sum_{i=p+1}^{q} \exp\big( \phi_{T_N-T_p^x}(v_{i-1}^x)
      + \scal{\psi_{T_N-T_p^x}(v_{i-1}^x)}{y} \big)  \notag\\
     &-\sum_{i=p+1}^{q} K_x \exp\big( \phi_{T_N-T_p^x}(u_i^x)
      + \scal{\psi_{T_N-T_p^x}(u_i^x)}{y} \big).
\end{align}
This function determines the exercise boundary for the price of the swaption. 

Now, since we cannot compute the characteristic function of $f(X_{T_p^x})$ 
explicitly, we will follow \citet{SingletonUmantsev02} and approximate $f$ by a 
linear function.

\begin{appS}
We approximate
\begin{align}\label{eq:ftilde-swaption}
f(X_{T_p^x}) \approx \widetilde{f}(X_{T_p^x})
  := \mathscr{A} + \scal{\mathscr{B}}{X_{T_p^x}},
\end{align}
where the constants $\mathscr{A}$, $\mathscr{B}$ are determined according to 
the linear regression procedure described in 
\citet[pp.~432-434]{SingletonUmantsev02}. The line 
$\scal{\mathscr{B}}{X_{T_p^x}}=-\mathscr{A}$ approximates the exercise boundary,
hence $\mathscr{A}$ and $\mathscr{B}$ are strike-dependent.
\end{appS}

The following assumption will be used for the pricing of swaptions and basis
swaptions.

\begin{assucd}
The cumulative distribution function of $X_t$ is continuous for all
$t\in[0,T_N]$.
\end{assucd}

Let $\Im(z)$ denote the imaginary part of a complex number $z \in \C$. Now, we
state the main result of this subsection.

\begin{proposition}
\label{prop:swaptionapprox}
Assume that $\mathscr{A,B}$ are determined by Approximation $(\mathbb{S})$ and 
that Assumption $(\mathbb{CD})$ is satisfied. The time-$0$ price of a payer 
swaption with strike $K$, option maturity $T_p^x$, and swap maturity $T_q^x$, is 
approximated by
\begin{align}\label{eq:swaptionapprox}
\widetilde{\mathbb{S}}^+_0(K,\mathcal{T}^x_{pq})
 &= B(0,T_N) \sum_{i=p+1}^q M_0^{v_{i-1}^x} \left[ \frac12 + \frac1\pi
    \int_0^\infty \frac{\Im\big(\widetilde{\xi}_{i-1}^x (z)\big)}{z} \dz \right]
     \notag\\
 &\qquad - K_x \sum_{i=p+1}^q B(0,T_i^x) \left[ \frac12 + \frac1\pi
    \int_0^\infty \frac{\Im\big(\widetilde{\zeta}_i^x (z)\big)}{z} \dz \right],
\end{align}
where $\widetilde{\zeta}_i^x$ and $\widetilde{\xi}_i^x$ are defined by
\eqref{eq:cfu} and \eqref{eq:cfv} respectively.
\end{proposition}

\begin{proof}
Starting from the swaption price in \eqref{eq:swaption-price-with-f} and using 
the relation between the terminal measure $\P_N$ and the measures $\P_k^x$ and 
$\overline{\P}_k^x$ in \eqref{Pkx-densities} and 
\eqref{Pkx-overline-densities}, we get that
\begin{multline}\label{eq:swaptiontrue}
\mathbb{S}^+_0(K,\mathcal{T}^x_{pq})
 = B(0,T_N) \sum_{i=p+1}^{q} M_0^{v_{i-1}^x}
    \bE_{i-1}^x \left[ \indik_{\{f(X_{T_p^x}) \geq 0\}} \right] \\
 - K_x \sum_{i=p+1}^{q} B(0,T_i^x) \,
   \E_i^x \left[ \indik_{\{f(X_{T_p^x}) \geq 0\}} \right].
\end{multline}
In addition, from the inversion formula of \citet{GilPelaez_1951} and using
Assumption $(\mathbb{CD})$, we get that
\begin{align}
\E_i^x \big[ \indik_{\{f(X_{T_p^x}) \geq 0\}} \big]
 = \frac12 + \frac1\pi \int_0^\infty \frac{\Im(\zeta_i^x (z))}{z} \dz,
   \label{eq:1-F(u)}\\
\bE_i^x \big[ \indik_{\{f(X_{T_p^x}) \geq 0\}} \big]
 = \frac12 + \frac1\pi \int_0^\infty \frac{\Im(\xi_i^x (z))}{z} \dz,
   \label{eq:1-F(v)}
\end{align}
for each $i\in\mathcal{K}^x$, where we define
\begin{align*}
\zeta_i^x (z) := \E_i^x \big[ \exp\big( \mathrm{i}z f(X_{T_p^x}) \big) \big]
\quad \text{and} \quad
\xi_i^x (z) := \bE_i^x \big[ \exp\big( \mathrm{i}z f(X_{T_p^x}) \big) \big].
\end{align*}

However, the above characteristic functions cannot be computed explicitly, in 
general, thus we will linearize the exercise boundary as described by 
Approximation $(\mathbb{S})$. That is, we approximate the unknown characteristic 
functions with ones that admit an explicit expression due to the affine property 
of $X$ under the forward measures. Indeed, using Approximation $(\mathbb{S})$, 
Proposition \ref{X-Pkx-characteristics} and Corollary 
\ref{X-Pkx-overline-characteristics} we get that
\begin{align}\label{eq:cfu}
\zeta_k^x (z) \approx \widetilde{\zeta}_k^x (z)
 &:= \E_k^x \big[ \exp\big( \mathrm{i}z \widetilde{f}(X_{T_p^x})\big) \big] \notag \\
 &\phantom{:}= \exp\left( \mathrm{i}z\mathscr{A}
      + \phi_{T_p^x}^{k,x}(\mathrm{i}z \mathscr{B})
      + \bscal{\psi_{T_p^x}^{k,x}(\mathrm{i}z \mathscr{B})}{X_0} \right),\\
\xi_k^x (z) \approx \widetilde{\xi}_k^x (z) \label{eq:cfv}
 &:= \bE_k^x \big[ \exp\big( \mathrm{i}z \widetilde{f}(X_{T_p^x})\big) \big] \notag\\
 &\phantom{:}= \exp\left(\mathrm{i}z\mathscr{A}
      + \bphi_{T_p^x}^{k,x}(\mathrm{i}z \mathscr{B})
      + \bscal{\bpsi_{T_p^x}^{k,x}(\mathrm{i}z \mathscr{B})}{X_0} \right).
\end{align}
After inserting \eqref{eq:1-F(u)} and \eqref{eq:1-F(v)} into
\eqref{eq:swaptiontrue} and using \eqref{eq:cfu} and \eqref{eq:cfv} we arrive at
the approximation formula for swaptions \eqref{eq:swaptionapprox}.
\end{proof}

\begin{remark}
The pricing of swaptions is inherently a high-dimensional problem. The
expectation in \eqref{val-swaption-general} corresponds to a $d$-dimensional
integral, where $d$ is the dimension of the driving process. However, the
exercise boundary is non-linear and hard to compute, in general. See, e.g.
\citet{BraceGatarekMusiela97}, \cite{EberleinKluge04} or \citet[\S 7.2, \S
8.3]{KellerResselPapapantoleonTeichmann09}  for some exceptional cases that
admit explicit solutions. Alternatively, one could express a swaption as a zero
strike basket option written on $2(q-p)$ underlying assets and use Fourier
methods for pricing; see \citet{HubalekKallsen03} or \citet{HurdZhou09}. This
leads to a $2(q-p)$-dimensional numerical integration. Instead, the
approximation derived in this section requires only the evaluation of $2(q-p)$
\emph{univariate} integrals together with the computation of the constants
$\mathscr{A,B}$. This reduces the complexity of the problem considerably.
\end{remark}

\subsection{Approximation formula for basis swaptions}
\label{section:Linear_boundary_approximation_BS}

In this subsection, we derive an analogous approximate pricing formula for basis
swaptions. Numerical results for this approximation will be reported in Section
\ref{sec:basis-swaption-ntest}.

Similar to the case of swaptions, we can rewrite the time-$0$ price of a basis 
swaption \eqref{val-bswaption-general} as follows:
\begin{equation}
\label{eq:basis-swaption-price-with-g}
\begin{multlined}
\mathbb{BS}^+_0(S,\mathcal{T}^{x_1}_{pq},\mathcal{T}^{x_2}_{pq}) =
 \\= B(0, T_N) \left\{ \sum_{i=p_2+1}^{q_2} \left(\E_{N} \left[
     M_{T_{p_2}^{x_2}}^{v_{i-1}^{x_2}} \indik_{\{g(X_{T_{p_2}}^{x_2})\ge0\}}\right] \! - \E_{N} \left[
     M_{T_{p_2}^{x_2}}^{u_{i}^{x_2}} \indik_{\{g(X_{T_{p_2}}^{x_2})\ge0\}}\right] \right)
     \right.
 \\- \left. \sum_{i=p_1+1}^{q_1} \left( \E_{N} \left[
      M_{T_{p_1}^{x_1}}^{v_{i-1}^{x_1}} \indik_{\{g(X_{T_{p_1}^{x_1}})\ge0\}}\right]
   - S_{x_1} \E_{N} \left[ M_{T_{p_1}^{x_1}}^{u_i^{x_1}}
        \indik_{\{g(X_{T_{p_1}^{x_1}})\ge0\}} \right] \right) \right\}, 
\end{multlined}
\end{equation}
where we define the function $g:\Rp^d\to\R$ by 
\begin{align}\label{eq:f-basis-swaption}
g(y)
 &= \sum_{i=p_2+1}^{q_2} \exp\big( \phi_{T_N-T_{p_2}^{x_2}}(v_{i-1}^{x_2})
   + \scal{\psi_{T_N-T_{p_2}^{x_2}}(v_{i-1}^{x_2})}{y} \big) \notag\\
 &\quad-   \sum_{i=p_2+1}^{q_2} \exp\big( \phi_{T_N-T_{p_2}^{x_2}}(u_{i}^{x_2})
   + \scal{\psi_{T_N-T_{p_2}^{x_2}}(u_{i}^{x_2})}{y} \big) \notag\\
 &\quad- \sum_{i=p_1+1}^{q_1} \exp\big( \phi_{T_N-T_{p_1}^{x_1}}(v_{i-1}^{x_1})
   + \scal{\psi_{T_N-T_{p_1}^{x_1}}(v_{i-1}^{x_1})}{y} \big) \notag\\
 &\quad+ \sum_{i=p_1+1}^{q_1} S_{x_1} \exp\big(
            \phi_{T_N-T_{p_1}^{x_1}}(u_i^{x_1})
          + \scal{\psi_{T_N-T_{p_1}^{x_1}}(u_i^{x_1})}{y} \big),
\end{align}
which determines the exercise boundary for the price of the basis swaption. This
will be approximated by a linear function following again
\citet{SingletonUmantsev02}.

\begin{appBS}
We approximate
\begin{align}\label{eq:ftilde-basis-swaption}
g(X_{T_{p_1}^{x_1}}) \approx \widetilde{g}(X_{T_{p_1}^{x_1}})
  := \mathscr{C} + \scal{\mathscr{D}}{X_{T_{p_1}^{x_1}}},
\end{align}
where $\mathscr{C}$ and $\mathscr{D}$ are determined via a linear regression.
\end{appBS}

\begin{proposition}
\label{eq:basis-swaptionapprox}
Assume that $\mathscr{C,D}$ are determined by Approximation $(\mathbb{BS})$ and
that Assumption $(\mathbb{CD})$ is satisfied. The time-$0$ price of a basis 
swaption with spread $S$, option maturity $T_{p_1}^{x_1}=T_{p_2}^{x_2}$, and 
swap maturity $T_{q_1}^{x_1}=T_{q_2}^{x_2}$, is approximated by
\begin{align}\label{eq:bswaptionapprox}
\widetilde{\mathbb{BS}}^+_0(S,\mathcal{T}^{x_1}_{pq},\mathcal{T}^{x_2}_{pq})
&= B(0,T_N) \sum_{i=p_2+1}^{q_2} M_0^{v_{i-1}^{x_2}} \left[ \frac12 +
    \frac1\pi \int_0^\infty \frac{\Im\big(\widetilde{\xi}_{i-1}^{x_2} (z)\big)}{z} \dz
    \right] \notag \\
    &\quad-  \sum_{i=p_2+1}^{q_2} B(0,T_i^{x_2}) \left[
    \frac12 + \frac1\pi \int_0^\infty
     \frac{\Im\big(\widetilde{\zeta}_i^{x_2} (z)\big)}{z} \dz \right] \nonumber \\
 &\!\!\!\!\! - B(0,T_N) \sum_{i=p_1+1}^{q_1} M_0^{v_{i-1}^{x_1}} \left[ \frac12
   + \frac1\pi \int_0^\infty \frac{\Im\big(\widetilde{\xi}_{i-1}^{x_1} (z)\big)}{z} \dz
    \right] \\ \nonumber 
 &\quad+ S_{x_1} \sum_{i=p_1+1}^{q_1} B(0,T_i^{x_1}) \left[
    \frac12 + \frac1\pi \int_0^\infty
     \frac{\Im\big(\widetilde{\zeta}_i^{x_1} (z)\big)}{z} \dz \right],
\end{align}
where $\widetilde\zeta_i^{x_l}$ and $\widetilde\xi_i^{x_l}$ are defined by
\eqref{eq:cfu-bs} and \eqref{eq:cfv-bs} for $l=1,2$.
\end{proposition}

\begin{proof}
Starting from the expression for the basis swaption price given in  
\eqref{eq:basis-swaption-price-with-g}, we follow the same steps as in the 
previous section: First, we use the relation between the terminal measure 
$\P_N$ and the measures $\P_k^x,\overline{\P}_k^x$ to arrive at an expression 
similar to \eqref{eq:swaptiontrue}. Second, we approximate $g$ by  
$\widetilde{g}$ in \eqref{eq:ftilde-basis-swaption}. Third, we define the 
approximate characteristic functions, which can be computed explicitly:
\begin{align}\label{eq:cfu-bs}
\widetilde{\zeta}_i^{x_l} (z)
 &:= \E_i^{x_l} \big[ \exp\big( \mathrm{i}z \widetilde{g}(X_{T_{p_l}^{x_l}})\big) \big] \notag \\
 &\phantom{:}= \exp\left( \mathrm{i}z\mathscr{C}
      + \phi_{T_{p_l}^{x_l}}^{i,x_l}(\mathrm{i}z \mathscr{D})
      + \bscal{\psi_{T_{p_l}^{x_l}}^{i,x_l}(\mathrm{i}z \mathscr{D})}{X_0} \right),\\
\widetilde{\xi}_i^{x_l} (z) \label{eq:cfv-bs}
 &:= \bE_i^{x_l} \big[ \exp\big( \mathrm{i}z \widetilde{g}(X_{T_{p_l}^{x_l}})\big) \big] \notag\\
 &\phantom{:}= \exp\left(\mathrm{i}z\mathscr{C}
      + \bphi_{T_{p_l}^{x_l}}^{i,x_l}(\mathrm{i}z \mathscr{D})
      + \bscal{\bpsi_{T_{p_l}^{x_l}}^{i,x_l}(\mathrm{i}z \mathscr{D})}{X_0} \right),
\end{align}
for $l=1,2$. Finally, putting all the pieces together we arrive at the
approximation formula \eqref{eq:bswaptionapprox} for the price of a basis
swaption.
\end{proof}
\section{Numerical examples and calibration}
\label{calibration}

The aim of this section is twofold: on the one hand, we demonstrate how the 
multiple curve \alm can be calibrated to market data and, on the other hand, we 
test the accuracy of the swaption and basis swaption approximation formulas. We 
start by discussing how to build a model which can simultaneously fit caplet 
volatilities when the options have different underlying tenors. Next, we test 
numerically the swaption and basis swaption approximation formulas 
\eqref{eq:swaptionapprox} and \eqref{eq:bswaptionapprox} using the calibrated 
models and parameters. In the last subsection, we build a simple model and 
compute exact and approximate swaption and basis swaption prices in a setup
which can be easily reproduced by interested readers.

\subsection{A specification with dependent rates}
\label{subsection:SingleTenor}

There are numerous ways of constructing models and the trade-off is usually 
between parsimony and fitting ability. We have elected here a heavily 
parametrized approach that focuses on the fitting ability, as we believe it 
best demonstrates the utility of our model. In particular, we want to show that 
\alms, which are driven by positive affine processes, can indeed be calibrated 
well to market data. Moreover, it is usually easier to move from a complex 
specification towards a simpler one, than the converse.

We provide below a model specification where LIBOR rates are driven by common 
and idiosyncratic factors which is suitable for sequential calibration to market 
data. The starting point is to revisit the expression for LIBOR rates in 
\eqref{eq:LIBOR-rate}:
\begin{align}\label{CalLIBOR}
1 + \delta_x L_k^x(t)
 &= M^{v_{k-1}^x}_t / M^{u_{k}^x}_t \\ \nonumber
 &\!\!\!\!\!\!\!\!\!\!\!\!\!\!\!\!\!\! =
    \exp\Big( \phi_{T_N-t}(v_{k-1}^x)-\phi_{T_N-t}(u_k^x) + \left\langle
     \psi_{T_N-t}(v_{k-1}^x)-\psi_{T_N-t}(u_k^x),X_t\right\rangle \Big).
\end{align}
According to Proposition \ref{initial-fit-multiple-curve}, when the dimension
of the driving process is greater than one, then the vectors $v_{k-1}^x$ and
$u_k^x$ are not fully determined by the initial term structure. Therefore, we
can navigate through different model specifications by altering the structure
of the sequences $(u_k^x)$ and $(v_k^x)$.

\begin{remark}\label{obsLFS}
The following observation allows to create an (exponential) linear factor
structure for the LIBOR rates with as many common and idiosyncratic factors as
desired. Consider an $\Rp^d$-valued affine process
\begin{align}
X &= (X^1,\dots, X^d),
\end{align}
and denote the vectors $v_{k-1}^x,u_k^x\in\Rp^d$ by
\begin{align}
v_{k-1}^x = (v_{1,k-1}^x,\dots,v_{d,k-1}^x)
 \quad\text{and}\quad
u_k^x = (u_{1,k}^x,\dots,u_{d,k}^x).
\end{align}
Select a subset $\mathscr{J}_k\subset\{1,\dots,d\}$, set $v_{i,k-1}^x=u_{i,k}^x$
for all $i\in\mathscr{J}_k$, and assume that $\{X^i\}_{i\in\mathscr{J}_k}$ are
independent of $\{X^j\}_{j\in\{1,\dots,d\}\backslash\mathscr{J}_k}$. Then, it
follows from \eqref{CalLIBOR} and \citet[Prop.~4.7]{KellerRessel08} that
$L_k^x$ will also be independent of $\{X^i\}_{i\in\mathscr{J}_k}$ and will
depend only on $\{X^j\}_{j\in\{1,\dots,d\}\backslash\mathscr{J}_k}$. The same 
observation allows also to construct a model where different factors are used 
for driving the OIS and LIBOR rates; see also Section \ref{subs:nrps}.
\end{remark}

Let $x_1,x_2\in\mathcal{X}$ and consider the tenor structures
$\mathcal{T}^{x_1},\mathcal{T}^{x_2}$ where
$\mathcal{T}^{x_2}\subset\mathcal{T}^{x_1}$. The dataset under consideration
contains caplets maturing on $M$ different dates for each tenor, where $M$ is
less than the number of tenor points in $\mathcal{T}^{x_1}$ and
$\mathcal{T}^{x_2}$. In other words,
only $M$ maturities are relevant for the calibration. The dynamics of OIS and
LIBOR rates are driven by tuples of affine processes
\begin{align}
\ud X^{i}_t
 & = -\lambda_{i} (X^{i}_t - \theta_{i}) \dt
   + 2 \eta_{i} \sqrt{X^{i}_t} \ud W^{i}_t
   + \ud Z^i_t,\label{eq:cir_pairs_1} \\
\ud X^{c}_t
 & = -\lambda_{c} (X^{c}_t - \theta_{c}) \dt
   + 2 \eta_{c} \sqrt{X^{c}_t} \ud W^{c}_t \label{eq:cir_pairs_2},
\end{align}
for $i=1,\dots,M$, where $X^c$ denotes the common and $X^i$ the idiosyncratic
factor for the $i$-th maturity. Here $X^i_0\in\Rp$,
$\lambda_{i},\theta_{i},\eta_{i}\in\Rp$ for $i=c,1,\dots,M,$ and
$W^{c},W^{1},\dots,W^{M},$ are independent Brownian motions. Moreover, $Z^i$ are
independent compound Poisson processes with constant intensity $\nu_i$ and
exponentially distributed jumps with mean values $\mu_i$, for $i=1,\dots,M$.
Therefore, the full process has dimension $M+1$:
\begin{align}
{X} = \left( X^{c}, X^{1}, \dots, X^{M} \right),
\end{align}
and the total number of process-specific parameters equals $5M+3$. The affine 
processes $X^c,X^1,\dots,X^M$ are mutually independent hence, using Proposition 
4.7 in \citet{KellerRessel08}, the functions $\phi_{(X^c,X^i)}$, respectively 
$\psi_{(X^{c},X^{i})}$, are known in terms of the functions $\phi_{X^c}$ and 
$\phi_{X^i}$, respectively $\psi_{X^c}$ and $\psi_{X^i}$, for all 
$i\in\{1,\dots,M\}$. The latter are provided, for example, by 
\citet[Ex.~2.3]{GrbacPapapantoleon13}.

\begin{figure}
\small
\centering
\begin{align*}
\hspace{-1.4em}
\begin{matrix}
\boxed{u_{\ell_1(M)}^{x_1}} &=& \big( \tilde{u}^{x_1}_{\ell_1(M)} &  0 & \dots
			    & 0 & 0 & 0
			    & \bar{u}^{x_1}_{\ell_1(M)} \big) \\
u_{\ell_1(M)-1}^{x_1} &=& \big( \tilde{u}^{x_1}_{\ell_1(M)-1} & 0 & \dots & 0 &
		      0 & 0
		      & \bar{u}^{x_1}_{\ell_1(M)-1} \big) \\
u_{\ell_1(M)-2}^{x_1} &=& \big( \tilde{u}^{x_1}_{\ell_1(M)-2} & 0 & \dots & 0 &
		      0 & 0
		      & \bar{u}^{x_1}_{\ell_1(M)-2} \big) \\
u_{\ell_1(M)-3}^{x_1} &=& \big( \tilde{u}^{x_1}_{\ell_1(M)-3} & 0 & \dots & 0 &
		      0 & 0
		      & \bar{u}^{x_1}_{\ell_1(M)-3} \big) \\
\boxed{u_{\ell_1(M-1)}^{x_1}} &=& \big( \tilde{u}^{x_1}_{\ell_1(M-1)} & 0 &
	\dots & 0 & 0
	&\bar{u}^{x_1}_{\ell_1(M-1)} &\bar{u}^{x_1}_{\ell_1(M)-3} \big) \\
u_{\ell_1(M-1)-1}^{x_1} &=& \big( \tilde{u}^{x_1}_{\ell_1(M-1)-1} & 0 & \dots &
	0 & 0
	& \bar{u}^{x_1}_{\ell_1(M-1)-1} &\bar{u}^{x_1}_{\ell_1(M)-3} \big) \\
u_{\ell_1(M-1)-2}^{x_1} &=& \big( \tilde{u}^{x_1}_{\ell_1(M-1)-2} & 0 & \dots &
	0 & 0
	& \bar{u}^{x_1}_{\ell_1(M-1)-2} & \bar{u}^{x_1}_{\ell_1(M)-3}\big) \\
u_{\ell_1(M-1)-3}^{x_1} &=& \big( \tilde{u}^{x_1}_{\ell_1(M-1)-3} & 0 & \dots &
	0 & 0
	& \bar{u}^{x_1}_{\ell_1(M-1)-3} & \bar{u}^{x_1}_{\ell_1(M)-3}\big) \\
\boxed{u_{\ell_1(M-2)}^{x_1}} &=& \big( \tilde{u}^{x_1}_{\ell_1(M-2)} & 0 &
	\dots & 0
	& \bar{u}^{x_1}_{\ell_1(M-2)} & \bar{u}^{x_1}_{\ell_1(M-1)-3}
	& \bar{u}^{x_1}_{\ell_1(M)-3}\big)\\
\vdots  && \vdots &  \vdots & & \Ddots & \vdots & \vdots & \vdots \\
\boxed{u_{\ell_1(1)}^{x_1}} &=& \big( \tilde{u}^{x_1}_{\ell_1(1)} &
	\bar{u}^{x_1}_{\ell_1(1)}
	& \bar{u}^{x_1}_{\ell_1(2)-3} & \dots & \bar{u}^{x_1}_{\ell_1(M-2)-3}
	& \bar{u}^{x_1}_{\ell_1(M-1)-3} & \bar{u}^{x_1}_{\ell_1(M)-3}\big)\\
\vdots  && \vdots & \vdots & \vdots & \vdots & \vdots & \vdots & \vdots \\
u_{1}^{x_1} &=& \big( \tilde{u}^{x_1}_{1} &\bar{u}^{x_1}_{1} &
\bar{u}^{x_1}_{\ell_1(2)-3}
	& \dots & \bar{u}^{x_1}_{\ell_1(M-2)-3} &\bar{u}^{x_1}_{\ell_1(M-1)-3}
	& \bar{u}^{x_1}_{\ell_1(M)-3}\big)
\end{matrix}
\end{align*}
\caption{The sequence $u^{x_1}$ encompasses the proposed `diagonal plus common'
factor structure. In this particular example, $x_1=3$ months and caplets
mature on entire years.}
\label{u-matrix-x1}
\end{figure}

In order to create a `diagonal plus common' factor structure, where each rate 
for each tenor is driven by an idiosyncratic factor $X^{i}$ and the common 
factor $X^{c}$, we will utilize Remark \ref{obsLFS}. We start from the longest 
maturity, which is driven by the idiosyncratic factor $X^M$ and the common 
factor $X^c$. Then, at the next caplet maturity date we add the independent 
idiosyncratic factor $X^{M-1}$, while we cancel the contribution of $X^M$ by 
`freezing' the values of $u^{x_1}$ and $v^{x_1}$ corresponding to that factor. 
The construction proceeds iteratively and the resulting structures for $u^{x_1}$ 
and $v^{x_1}$ are presented in Figures \ref{u-matrix-x1} and \ref{v-matrix-x1}, 
where elements of $u^{x_1}$ below a certain `diagonal' are `frozen' to the 
latest-set value and copied to $v^{x_1}$. These structures produce the desired 
feature that each rate is driven by an idiosyncratic and a common factor, while 
they do not violate inequalities $v_k^x\ge u_k^x\ge u_{k+1}^x$, that stemm from 
Propositions \ref{initial-fit-multiple-curve-OIS} and 
\ref{initial-fit-multiple-curve}. In Figures \ref{u-matrix-x1} and 
\ref{v-matrix-x1}, $\ell_1(k):=k/\delta_{x_1}$ for $k=1,\dots,M$, i.e. this 
function maps caplet maturities into tenor points. Moreover, all elements in 
these matrices are non-negative and 
$u^{x_1}_{N^{x_1}}=v^{x_1}_{N^{x_1}}=0\in\R^{M+1}$.

\begin{figure}
\small
\begin{align*}
\hspace{-1.5em}
\begin{matrix}
v_{\ell_1(M)}^{x_1} &=& \big( \tilde{v}_{\ell_1(M)}^{x_1}& 0  & \dots & 0 & 0
		    & 0 & \bar{v}^{x_1}_{\ell_1(M)} \big) \\	
\boxed{v_{\ell_1(M)-1}^{x_1}} &=& \big( \tilde{v}_{\ell_1(M)-1}^{x_1} & 0
		    & \dots & 0 & 0  & 0 & \bar{v}^{x_1}_{\ell_1(M)-1}\big)\\
v_{\ell_1(M)-2}^{x_1} &=& \big( \tilde{v}_{\ell_1(M)-2}^{x_1} & 0 &\dots& 0& 0
		    & 0  & \bar{v}^{x_1}_{\ell_1(M)-2} \big)\\
v_{\ell_1(M)-3}^{x_1} &=& \big( \tilde{v}_{\ell_1(M)-3}^{x_1} & 0 & \dots&0&0
		    & 0 & \bar{v}^{x_1}_{\ell_1(M)-3} \big)\\
v_{\ell_1(M-1)}^{x_1} &=& \big( \tilde{v}_{\ell_1(M-1)}^{x_1} &0& \dots&0  & 0
		    &\bar{v}^{x_1}_{\ell_1(M-1)}
		    & \bar{u}^{x_1}_{\ell_1(M)-3}\big)\\
\boxed{v_{\ell_1(M-1)-1}^{x_1}} &=& \big( \tilde{v}_{\ell_1(M-1)-1}^{x_1} & 0
    & \dots & 0 & 0  & \bar{v}^{x_1}_{\ell_1(M-1)-1}
    & \bar{u}^{x_1}_{\ell_1(M)-3}\big)\\
v_{\ell_1(M-1)-2}^{x_1} &=& \big( \tilde{v}_{\ell_1(M-1)-2}^{x_1} & 0
    & \dots & 0 & 0  & \bar{v}^{x_1}_{\ell_1(M-1)-2}
    & \bar{u}^{x_1}_{\ell_1(M)-3}\big)\\
v_{\ell_1(M-1)-3}^{x_1} &=& \big( \tilde{v}_{\ell_1(M-1)-3}^{x_1} & 0
  & \dots & 0 & 0
  & \bar{v}^{x_1}_{\ell_1(M-1)-3} & \bar{u}^{x_1}_{\ell_1(M)-3}\big)\\
v_{\ell_1(M-2)}^{x_1} &=& \big( \tilde{v}_{\ell_1(M-2)}^{x_1} & 0  & \dots & 0
  & \bar{v}^{x_1}_{\ell_1(M-2)}
  & \bar{u}^{x_1}_{\ell_1(M-1)-3}
  & \bar{u}^{x_1}_{\ell_1(M)-3}\big)\\
\boxed{v_{\ell_1(M-2)-1}^{x_1}} &=& \big( \tilde{v}_{\ell_1(M-2)-1}^{x_1} & 0
  & \dots & 0  & \bar{v}^{x_1}_{\ell_1(M-2)-1}
  & \bar{u}^{x_1}_{\ell_1(M-1)-3}
  & \bar{u}^{x_1}_{\ell_1(M)-3}\big)\\
\vdots && \vdots &  \vdots & & \Ddots & \vdots & \vdots &\vdots \\
v_{\ell_1(1)}^{x_1} &=& \big( \tilde{v}_{\ell_1(1)}^{x_1}
  & \bar{v}^{x_1}_{\ell_1(1)} & \bar{u}^{x_1}_{\ell_1(2)-3} & \dots
  & \bar{u}^{x_1}_{\ell_1(M-2)-3}
  & \bar{u}^{x_1}_{\ell_1(M-1)-3}
  & \bar{u}^{x_1}_{\ell_1(M)-3}\big)\\
\boxed{v_{\ell_1(1)-1}^{x_1}} &=& \big( \tilde{v}_{\ell_1(1)-1}^{x_1}
  & \bar{u}^{x_1}_{\ell_1(1)-3} & \bar{u}^{x_1}_{\ell_1(2)-3} & \dots
  & \bar{u}^{x_1}_{\ell_1(M-2)-3}
  & \bar{u}^{x_1}_{\ell_1(M-1)-3}
  & \bar{u}^{x_1}_{\ell_1(M)-3}\big)\\
\vdots & & \vdots & \vdots & \vdots & \vdots & \vdots & \vdots & \vdots\\
v_{1}^{x_1} &=& \big( \tilde{v}_{1}^{x_1} & \bar{u}^{x_1}_{\ell_1(1)-3}
  & \bar{u}^{x_1}_{\ell_1(2)-3} & \dots   & \bar{u}^{x_1}_{\ell_1(M-2)-3}
  & \bar{u}^{x_1}_{\ell_1(M-1)-3} & \bar{u}^{x_1}_{\ell_1(M)-3}\big)
\end{matrix}
\end{align*}
\caption{The sequence $v^{x_1}$ is constructed analogously to $u^{x_1}$. In this
particular example, $x_1=3$ months and caplets mature on entire years.}
\label{v-matrix-x1}
\end{figure}

The boxed elements are the only ones that matter in terms of pricing caplets
when these are not available at every tenor date of $\mathcal{T}^{x_1}$. The
impact of the common factor is determined by the difference between
$\tilde{v}_{k-1}^{x_1}$ and $\tilde{u}_{k}^{x_1}$. If we set
$\tilde{v}_{k-1}^{x_1}=\tilde{u}_{k}^{x_1}$, it follows from Remark \ref{obsLFS}
that $L_k^{x_1}$ will be independent of the common factor $X^c$ and thus
determined solely by the corresponding idiosyncratic factor $X^i$, with
$k=\ell_1(i)$. If  the values of $\tilde{v}_{k}^{x_1}$ and
$\tilde{u}_{k}^{x_1}$ are fixed a priori, the remaining values
$(\bar{u}_k^{x_1})_{k=1,\dots,N^{x_1}}$ and
$(\bar{v}_k^{x_1})_{k=1,\dots,N^{x_1}}$ are determined uniquely by the initial
term structure of OIS and LIBOR rates; see again Propositions
\ref{initial-fit-multiple-curve-OIS} and \ref{initial-fit-multiple-curve}. This
model structure is consistent with $v_{k-1}^{x_1} \ge u_{k-1}^{x_1} \ge
u_k^{x_1}$ if and only if the sequences $\tilde{u}^{x_1}$ and $\bar{u}^{x_1}$
are decreasing, $\tilde{v}_k^{x_1}\ge\tilde{u}_k^{x_1}$ and
$\bar{v}_k^{x_1}\ge\bar{u}_k^{x_1}$ for every $k=1,\dots,N^{x_1}$. Moreover,
this structure will be consistent with the `normal' market situation
described in Remark \ref{rem:relation-u-v} if, in addition,
$\tilde{v}^{x_1}_k\in[\tilde{u}_k^{x_1},\tilde{u}_{k-1}^{x_1}]$ and
$\bar{v}_k^{x_1}\in[\bar{u}_k^{x_1},\bar{u}_{k-1}^{x_1}]$ for every
$k=1,\dots,N^{x_1}$.

The corresponding matrices for the $x_2$ tenor are constructed in a similar 
manner. More precisely, $u^{x_2}$ is constructed by simply copying the relevant 
rows from $u^{x_1}$. Simultaneously, for $v^{x_2}$ the elements 
$(\bar{v}_k^{x_2})_{k=0,\dots,N^{x_2}}$ are introduced in order to fit the 
$x_2$-initial LIBOR term structure, as well as the elements 
$(\tilde{v}_k^{x_2})_{k=0,\dots,N^{x_2}}$ which determine the role of the 
common factor. We present only four rows from these matrices in Figures 
\ref{u-matrix-x2} and \ref{v-matrix-x2}, for the sake of brevity.

\begin{figure}
\centering
\begin{align*}
\begin{matrix}
u_{\ell_2(M)}^{x_2} &=& \big( \tilde{u}^{x_1}_{\ell_1(M)} &   0 & \dots & 0 & 0
		    &\bar{u}^{x_1}_{\ell_1(M)} \big) \\
u_{\ell_2(M)-1}^{x_2} &=& \big( \tilde{u}^{x_1}_{\ell_1(M)-2} & 0 & \dots & 0
	      & 0 &\bar{u}^{x_1}_{\ell_1(M)-2} \big)\\
u_{\ell_2(M-1)}^{x_2} &=& \big( \tilde{u}^{x_1}_{\ell_1(M-1)} & 0 & \dots & 0
	      &\bar{u}^{x_1}_{\ell_1(M-1)} &\bar{u}^{x_1}_{\ell_1(M)-3}\big)\\
u_{\ell_2(M-1)-1}^{x_2}&=& \big( \tilde{u}^{x_1}_{\ell_1(M-1)-2}  & 0 & \dots
	      & 0  & \bar{u}^{x_1}_{\ell_1(M-1)-2}
	      & \bar{u}^{x_1}_{\ell_1(M)-3}\big)
\end{matrix}
\end{align*}
\caption{The first four rows of $u^{x_2}$. In this particular example, $x_2=6$
months and caplets mature on entire years.}
\label{u-matrix-x2}
\begin{align*}
\begin{matrix}
v_{\ell_2(M)}^{x_2}  &=& \big( \tilde{v}_{\ell_2(M)}^{x_2}&   0 & \dots & 0
& 0 & \bar{v}^{x_2}_{\ell_2(M)} \big) \\
v_{\ell_2(M)-1}^{x_2} &=& \big( \tilde{v}_{\ell_2(M)-1}^{x_2}& 0 & \dots  &
0 & 0 & \bar{v}^{x_2}_{\ell_2(M)-1} \big)\\
v_{\ell_2(M-1)}^{x_2} &=& \big( \tilde{v}_{\ell_2(M-1)}^{x_2}& 0 & \dots &
0 &\bar{v}^{x_2}_{\ell_2(M-1)} &\bar{u}^{x_1}_{\ell_1(M)-3}\big)\\
v_{\ell_2(M-1)-1}^{x_2} &=& \big( \tilde{v}_{\ell_2(M-1)-1}^{x_2} & 0 &
\dots &  0 & \bar{v}^{x_2}_{\ell_2(M-1)-1} &
\bar{u}^{x_1}_{\ell_1(M)-3}\big)
\end{matrix}
\end{align*}
\caption{The first four rows of $v^{x_2}$. In this particular example, $x_2=6$
months and caplets mature on entire years.}
\label{v-matrix-x2}
\end{figure}

\subsection{Calibration to caplet data}

The data we use for calibration are from the EUR market on 27 May 2013 
collected from Bloomberg. These market data correspond to fully collateralized 
contracts, hence they are considered `clean'. Bloomberg provides synthetic zero 
coupon bond prices for EURIBOR rates, as well as OIS rates constructed in a 
manner described in \citet{Akkara_2012}. In our example, we will focus on the 3 
and 6 month tenors only. The zero coupon bond prices are converted into zero 
coupon rates and plotted in Figure \ref{fig-InitialTS}.
\begin{figure}
 \includegraphics[scale=0.5]{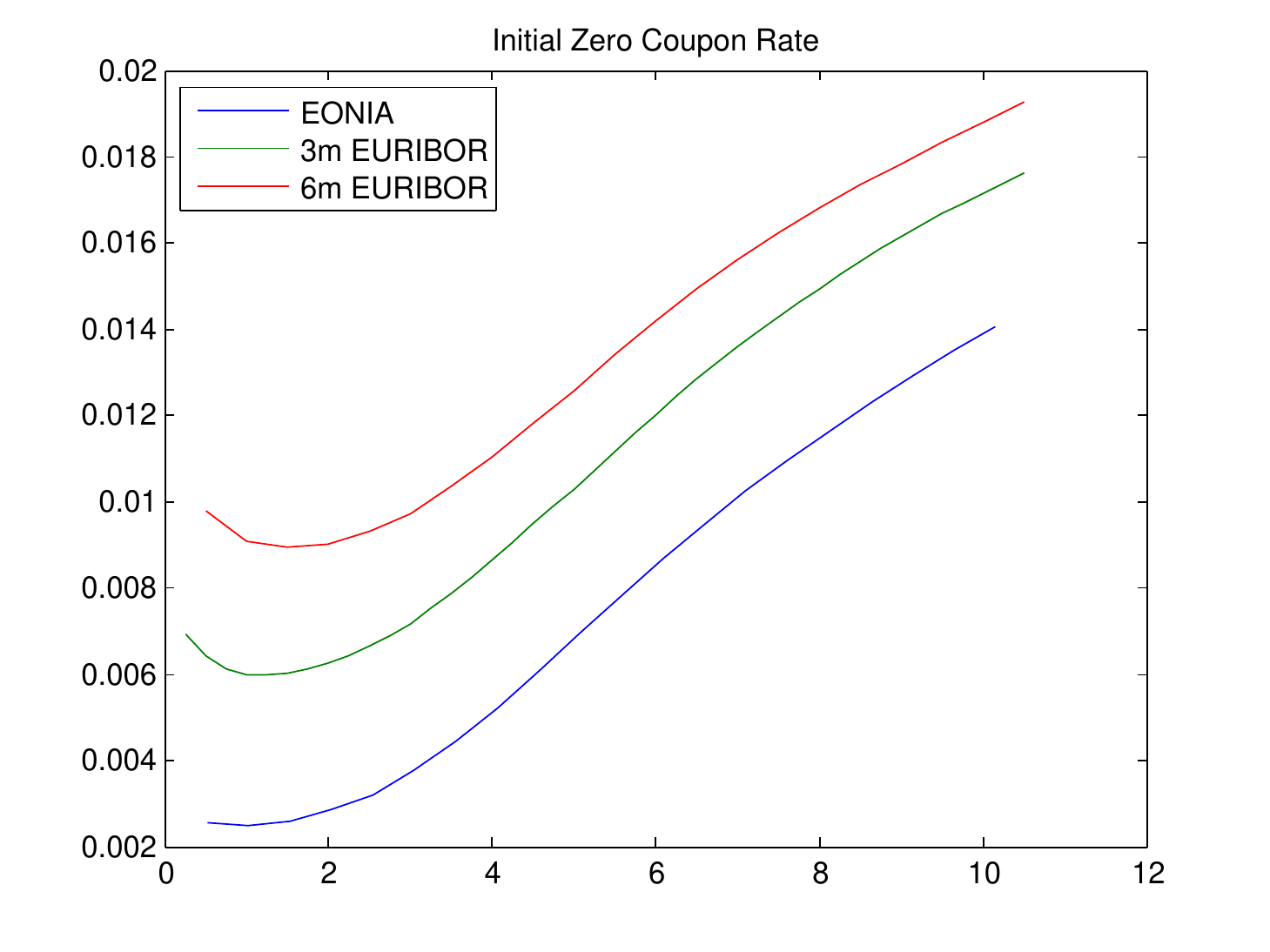}
 \vspace{-.75em}
 \caption{Zero coupon rates, EUR market, 27 May 2013.}
 \label{fig-InitialTS}
\end{figure}
Cap prices are converted into caplet implied volatilities using the algorithm 
described in \citet{Levin_2012}. The implied volatility is calculated using OIS 
discounting when inverting the \citet{Black76} formula. The caplet data we have 
at our disposal correspond to 3- and 6-month tenor structures. More precisely, 
in the EUR market caps written on the 3-month tenor are quoted only up to a 
maturity of 2 years, while 6-month tenor caps are quoted from maturity 3 years 
and onwards. Moreover, we have option prices only for the maturities 
corresponding to entire years and not for every tenor point. We have a grid of 
14 strikes ranging from $1\%$ to $10\%$ as quoted in Bloomberg. We calibrate to 
caplet data for maturities $1,2,\dots,10$ years and the OIS zero coupon bond
$B(\cdot,10.5)$ defines the terminal measure\footnote{We found that the model
performs slightly better in calibration using this numeraire, than when choosing
10 years as the terminal maturity.}. We fix in advance the values of the
parameters $(\tilde{u}^{x_1}), (\tilde{u}^{x_2})$, $(\tilde{v}^{x_1})$ and 
$(\tilde{v}^{x_2})$, as well as the parameters of the process $X^c$. The impact 
of $X^c$ is determined by the spread between $\tilde{u}^{x_1}_{\ell_1(i)}$ and 
$\tilde{v}^{x_1}_{\ell_1(i)-1}$, and $\tilde{u}^{x_2}_{\ell_2(i)}$ and 
$\tilde{v}^{x_2}_{\ell_2(i)-1}$ for the 3m and 6m tenor caplets respectively, 
and we will simplify by setting 
$\tilde{u}^{x_1}_1=\dots=\tilde{u}^{x_1}_{N^{x_1}-1}=u_c$ constant. The 
constant $u_c$, along with $\tilde{v}^{x_1}_1,\dots,\tilde{v}^{x_1}_{N^{x_1}}$ 
and $\tilde{v}^{x_2}_1,\dots,\tilde{v}^{x_2}_{N^{x_2}}$ are not identified by 
the initial term structures and have to be determined in some other manner 
(e.g. by calibration to swaptions or basis swaptions). They also cannot be 
chosen completely freely and one has to validate that the values of 
$u_k^{x_1},u_k^{x_2}$ and $v_k^{x_1},v_k^{x_2}$ stemming from these procedures 
satisfy the necessary inequalities, i.e. $v^x_{k-1}\ge u^x_{k-1}\ge u^x_k$. 
Having this in mind, we chose these values in a manner such that $X^c$ accounts 
for approximately 50\% of the total variance of LIBOR rates from maturities 4 
until 10, and about 10\% of the total variance for maturities 1 until 3. We have 
verified through experimentation that this ad-hoc choice of dependence structure 
does not have a qualitative impact on the results of the following sections. 
Alternatively, these parameters could be calibrated to derivatives such as 
swaptions,  basis swaptions or other derivatives partly determined by the 
dependence structure of the LIBOR rates. However, since interest rate 
derivative markets remain highly segmented and joint calibration of caplets and 
swaptions is a perennial challenge, see e.g. \cite{BrigoMercurio06} or 
\citet*{Ladkau_Schoenmakers_Zhang_2013}, we will leave this issue for future 
research.

\begin{figure}[ht!]
 \centering
  \includegraphics[width=12.5cm]{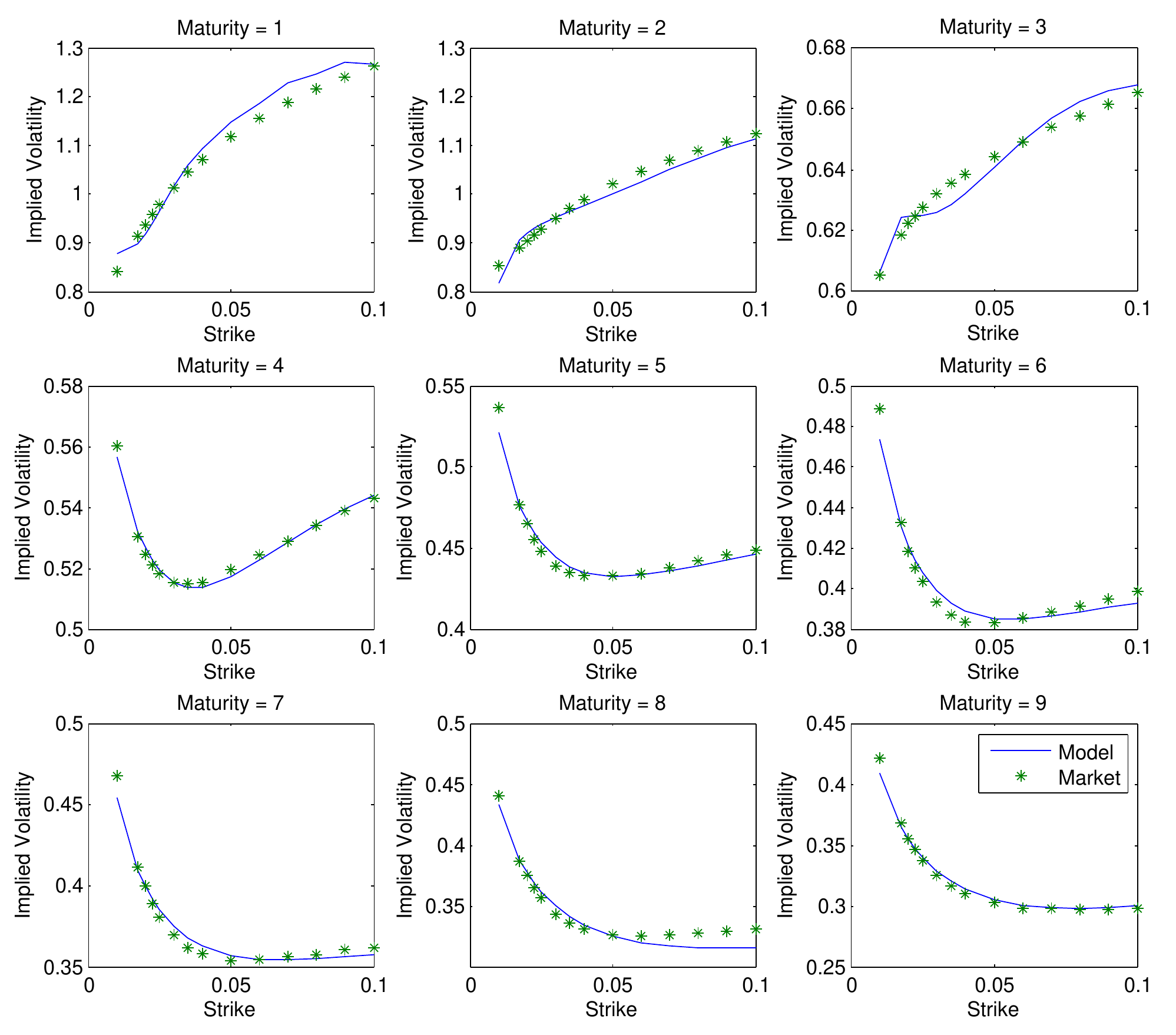}
 \vspace{-.5em}
\caption{Market and model implied volatility for caplets written on 3 (1--2
year maturity) and 6 (2--9 year maturity) month tenor EURIBOR.}
\label{figure:caplets3m}
\end{figure}

The model construction summarized in Figures 
\ref{u-matrix-x1}--\ref{v-matrix-x2} has the advantage that caplets can be 
calibrated sequentially one maturity at a time starting at the longest maturity 
and then moving backwards. In the calibration procedure we fit the parameters 
of each idiosyncratic process $X^{i}$ to caplet prices with maturity $T_i^x$ 
while simultaneously choosing $u_i^{x_1},u_i^{x_2}$ and $v_i^{x_1},v_i^{x_2}$ 
to match the corresponding values of the initial OIS and LIBOR rates. Caplets 
are priced using formula \eqref{MC-ALM-caplets} and the parameters are found 
using standard least-squares minimization between market and model implied 
volatility.  The results\footnote{All calibrated and chosen parameter values as 
well as the calibrated matrices $u^{x_j},v^{x_j}$ for $j=1,2$ are available 
from the authors upon request.} from fitting the caplets are shown in Figure 
\ref{figure:caplets3m}. We can observe that the model performs very well for 
different types of volatility smiles across the whole term structure, with only 
minor problems for extreme strikes in maturities 1-3. These problems are 
however mainly cosmetic in nature as these prices and more importantly the 
deltas for these contracts are very close to zero anyway, making any model 
error in this region economically insignificant.

\subsection{Swaption price approximation}
\label{sec:swaption-ntest}

The next two sections are devoted to numerically testing the validity of the
swaption and basis swaption price approximation formulas derived in Sections
\ref{section:Linear_boundary_approximation} and
\ref{section:Linear_boundary_approximation_BS}. We will run a Monte Carlo study
comparing the true price with the linear boundary approximation formula. The
model parameters used stem from the calibration to the market data described in
the previous section.

Let us denote the true and the approximate prices as follows:
\begin{align*}
\mathbb{S}^{+}_0 (K,\mathcal{T}^x_{pq})
 &= B(0, T_N) \, \E_N \left[  \left( \sum_{i=p+1}^{q}
     M_{T_p^x}^{v_{i-1}^x} - \sum_{i=p+1}^{q} K_x M_{T_p^x}^{u_{i}^x}
     \right) \indik_{\{f(X_{T_p^x}) \geq 0\}}\right], \\
\widetilde{\mathbb{S}}^{+}_0 (K,\mathcal{T}^x_{pq})
 &= B(0, T_N) \, \E_N \left[  \left( \sum_{i=p+1}^{q}
    M_{T_p^x}^{v_{i-1}^x} - \sum_{i=p+1}^{q} K_x M_{T_p^x}^{u_{i}^x}
    \right) \indik_{\{\widetilde{f}(X_{T_p^x}) \geq 0\}} \right],
\end{align*}
where $f$ and $\widetilde{f}$ were defined in \eqref{eq:f_swaption} and
\eqref{eq:ftilde-swaption} respectively. The Monte Carlo (MC)
estimator\footnote{We construct the Monte Carlo estimate  using 5 million paths
of $X$ with 10 discretization steps per year. In each discretization step the
continuous part is simulated using the algorithm in \citet[\S
3.4.1]{Glasserman03} while the jump part is handled using \citet[pp.
137--139]{Glasserman03} with jump size distribution changed from log-normal to
exponential.} of $\mathbb{S}^{+}_0(K,\mathcal{T}^x_{pq})$ is denoted by
$\hat{\mathbb{S}}^{+}_0(K,\mathcal{T}^x_{pq})$ and we will refer to it as the
`true price'. Instead of computing
$\widetilde{\mathbb{S}}^{+}_0(K,\mathcal{T}^x_{pq})$ using Fourier methods, we
will form another MC estimator
$\hat{\widetilde{\mathbb{S}}}^{+}_0(K,\mathcal{T}^x_{pq})$. This has the
advantage that, when the same realizations are used to calculate both MC
estimators, the difference $\hat{{\mathbb{S}}}^{+}_0(K,\mathcal{T}^x_{pq})-
\hat{\widetilde{\mathbb{S}}}^{+}_0(K,\mathcal{T}^x_{pq})$ will be an estimate of
the error induced by the linear boundary approximation which is minimally
affected by simulation bias.

Swaption prices vary considerably across strike and maturity, thus we will
express the difference between the true and the approximate price in terms of
implied volatility (using OIS discounting), which better demonstrates the
economic significance of any potential errors. We price swaptions on three
different underlying swaps. The results for the 3m underlying tenor are
exhibited in Figure \ref{figure:Swaptiontest}. The corresponding results for the
6m tenor swaptions have errors which are approximately one half the level in the
graphs shown here and have been omitted for brevity.

On the left hand side of Figure \ref{figure:Swaptiontest}, implied volatility
levels are plotted for the true and the approximate prices. The strikes are
chosen to range from 60\% to 200\% of the spot value of the underlying fair swap
rate, which is the normal range the products are quoted. The right hand side
shows the difference between the two implied volatilities in basis points (i.e.
multiplied by $10^4$). As was also documented in \citet{Schrager_Pelsser_2006},
the errors of the approximation usually increase with the number of payments in
the underlying swap. This is also the case here, however the level of the errors
is in all cases very low. In normal markets, bid-ask spreads typically range
from 10 to 300 bp (at the at-the-money level) thus even the highest errors are
too small to be of any economic significance. This is true even in the case of
the 2Y8Y swaption which contains 32 payments.

\begin{figure}
 \centering
\includegraphics[width=10.5cm]{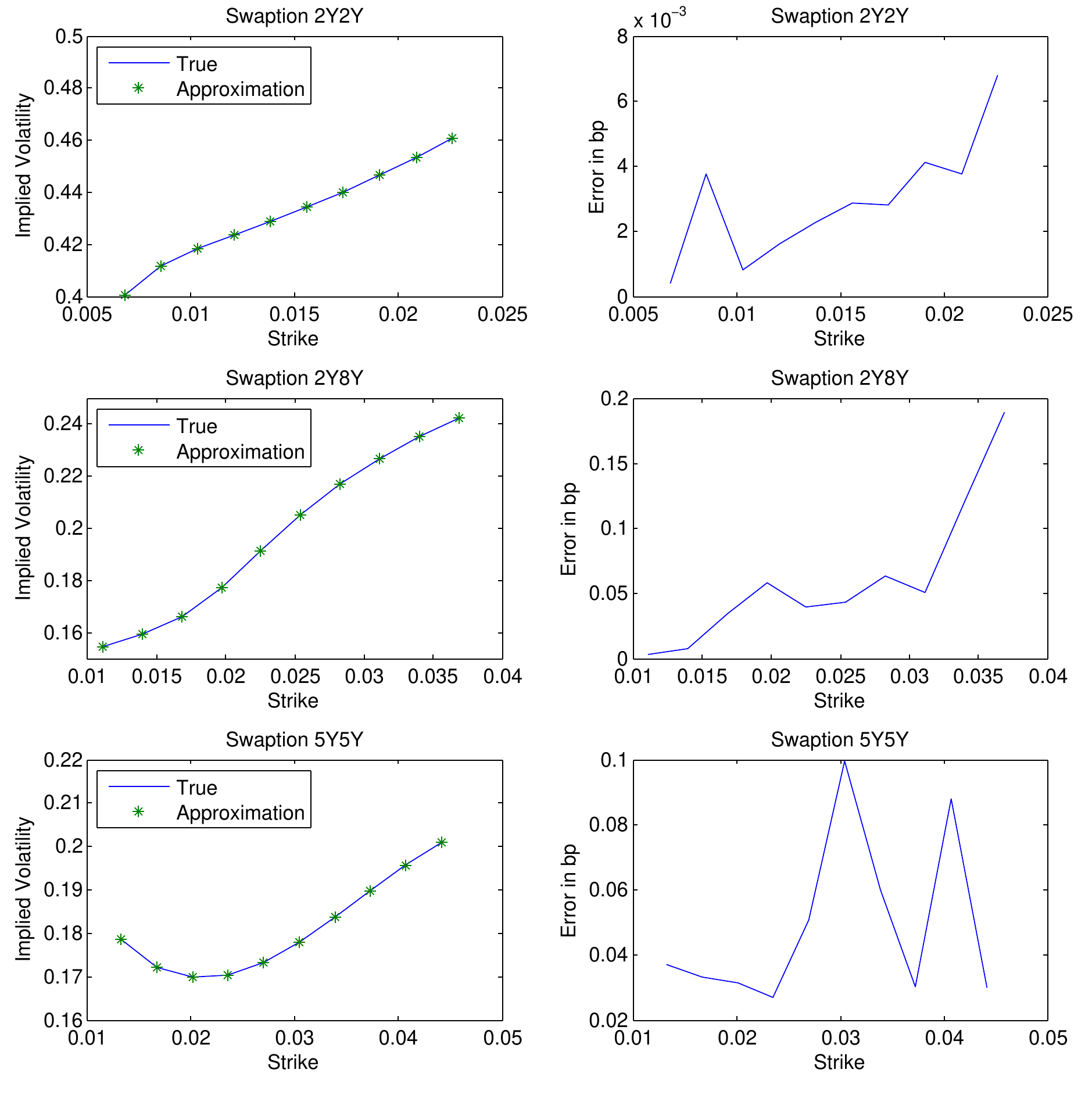}
 \vspace{-1.em}
 \caption{Implied volatility and absolute errors for 3m swaptions.}
 \label{figure:Swaptiontest}
\end{figure}

\subsection{Basis swaption price approximation}
\label{sec:basis-swaption-ntest}

In order to test the approximation formula for basis swaptions, we will follow
the same methodology as in the previous subsection. That is, we calculate MC
estimators for the following two expectations:
\begin{multline*}
\mathbb{BS}^{+}_0 (S,\mathcal{T}^{x_1}_{pq} ,\mathcal{T}^{x_2}_{pq})
 = B(0,T_N) \, \E_{N} \Bigg[ \Bigg( \sum_{i=p_2+1}^{q_2} \left(
    M_{T_{p_2}^{x_2}}^{v_{i-1}^{x_2}} - M_{T_{p_2}^{x_2}}^{u_{i}^{x_2}} \right)
 \\ - \sum_{i=p_1+1}^{q_1} \left( M_{T_{p_1}^{x_1}}^{v_{i-1}^{x_1}}
    - S_{x_1} M_{T_{p_1}^{x_1}}^{u_i^{x_1}} \right) \Bigg)
     \indik_{\{g(X_{T_{p_2}})\ge0\}} \Bigg],
\end{multline*}

\begin{multline*}
\widetilde{\mathbb{BS}}^{+}_0 (S,\mathcal{T}^{x_1}_{pq},\mathcal{T}^{x_2}_{pq})
 = B(0,T_N) \, \E_{N} \Bigg[ \Bigg( \sum_{i=p_2+1}^{q_2} \left(
    M_{T_{p_2}^{x_2}}^{v_{i-1}^{x_2}} - M_{T_{p_2}^{x_2}}^{u_{i}^{x_2}} \right)
 \\ - \sum_{i=p_1+1}^{q_1} \left( M_{T_{p_1}^{x_1}}^{v_{i-1}^{x_1}}
    - S_{x_1} M_{T_{p_1}^{x_1}}^{u_i^{x_1}} \right) \Bigg)
    \indik_{\{\widetilde{g}(X_{T_{p_2}})\ge0\}} \Bigg],
\end{multline*}
where $S_{x_1}=1-\delta_{x_1}S$, while $g$ and $\widetilde{g}$ were defined in
\eqref{eq:f-basis-swaption}  and \eqref{eq:ftilde-basis-swaption}. Using the
same realizations, we plot the level, absolute and relative differences between
both prices measured in basis points as a function of the spread for three
different underlying basis swaps. The spreads are chosen to range from 50\% to
200\% of the at-the-money level, i.e. the spread  that sets the underlying
basis swap to a value of zero, see again \eqref{eq:basis-swap-spread},
\begin{align*}
S_{ATM} := S_0(\mathcal{T}^{x_1}_{pq},\mathcal{T}^{x_2}_{pq}).
\end{align*}

\begin{figure}
 \centering
  \includegraphics[width=12.75cm]{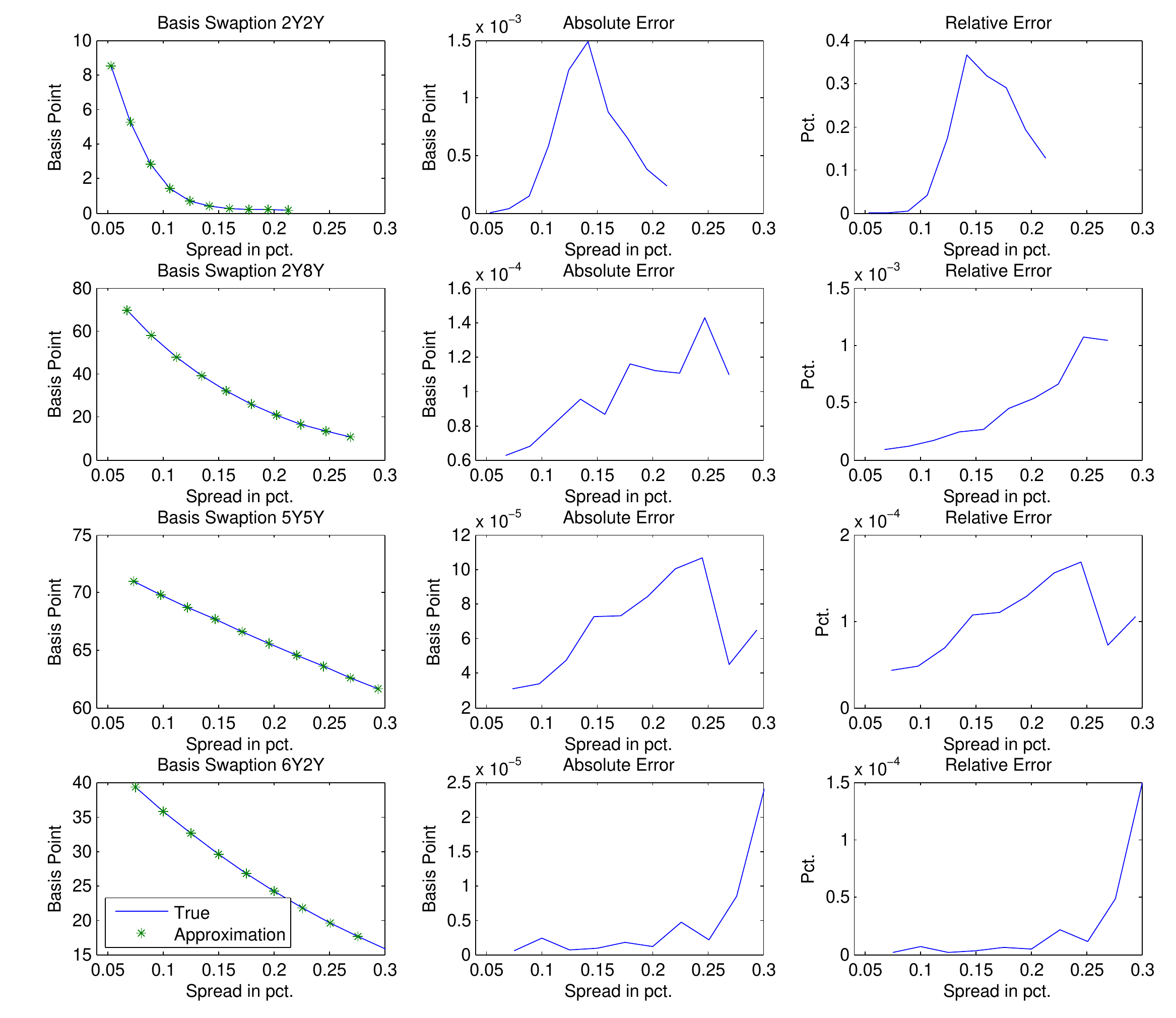}
 \vspace{-1.em}
 \caption{Prices in basis points, absolute and relative errors for 3m--6m basis
swaptions.}
 \label{figure:BasisSwaptiontest}
\end{figure}

The numerical results can be seen in Figure \ref{figure:BasisSwaptiontest}. We
have chosen these maturities to be representative of two general patterns. The
first is that the errors tend to increase with the length of the basis swap,
which is exemplified by comparing errors for the 2Y8Y, 5Y5Y and 6Y2Y contracts.
The second pattern relates to when the majority of payments in the contract are
paid out. We can notice that the errors for the 2Y2Y contract are much larger
than for the corresponding 6Y2Y, even though both contain the same number of
payments. Furthermore, we can also see that the 2Y2Y contract has larger errors
than the 2Y8Y even though both have the same maturity and the latter has more
payments. This anomalous result can be explained by the convexity of the term
structure of interest rates. In Figure \ref{fig-InitialTS} we can notice that
the majority of the payments of the 2Y2Y contract fall in a particularly curved
region of the term structure. This will result in an exercise boundary which is
also more nonlinear, thus leading to the relative deterioration of the linear
boundary approximation. However, it must be emphasized that the errors are still
at a level easily deemed economically insignificant, with a maximum relative
error of 0.4\% in a spread region where the price levels are particularly low.

\begin{remark}
The approximative formulas from Propositions \ref{prop:swaptionapprox} and 
\ref{eq:basis-swaptionapprox} can be used for calibration to swaptions and 
basis swaptions. Error bounds for these approximations are not available in 
closed form and thus accuracy in the entire parameter space cannot be 
guaranteed. Any desired accuracy in pricing is achievable using Monte Carlo 
methods, which means that the accuracy of the approximate formula can always be 
validated numerically. However, when performing a full calibration, which may 
require several hundreds of iterations in order to achieve convergence in a 
numerical optimization procedure, Monte Carlo methods are slow in comparison to
the analytical approximation. Thus, in order to calibrate to swaption prices, 
one would prefer the approximation to the Monte Carlo method. Then one can 
perform a Monte Carlo simulation (just one) to validate that the approximation 
is also correct for the parameters found by the numerical optimization. Let us 
also mention that in a typical calibration procedure an acceptable error is 
around 2\%, well above the error of the approximative formulas, else the risk 
of overifitting the data is present.
\end{remark}

\subsection{A simple example}

The purpose of this section is to provide a simple example to help the reader's 
intuition regarding the numerical implementation of the model. We present a 
fully constructed and more manageable numerical toy example of fitting the 
model parameters $u^x_k$ and $v^x_k$ to the initial term structures, which can 
be reproduced by the reader himself/herself as opposed to the full calibration 
example in Sections \ref{sec:swaption-ntest} and 
\ref{sec:basis-swaption-ntest}. Moreover, we show in this simple setting how  
Approximations $(\mathbb{S})$ and $(\mathbb{BS})$ are computed.

We start by choosing a simple two factor model $X=(X^1,X^2)$ with
\begin{align}
\ud X^{i}_t
 & = -\lambda_{i} (X^{i}_t - \theta_{i}) \dt
   + 2 \eta_{i} \sqrt{X^{i}_t} \ud W^{i}_t
   + \ud Z^i_t,\label{eq:cir_a} \quad i=1,2,
\end{align}
where we set
\begin{center}
\begin{tabular}{c|cccccc}
$i$&$X_0^i$ & $\lambda_i$ & $\theta_i$ &$\eta_i$ &$\nu_i$& $\mu_i$ \\
\hline
1 &  0.5000  &  0.1000  &  1.5300  &  0.2660  &  0       &  0\\
2 &  9.4531  &  0.0407  &  0.0591  &  0.4640  &  0.0074  &  0.2499\\
\end{tabular}\\
\end{center}
The initial term structures are constructed from a Nelson--Siegel
parametrization of the zero coupon rate $R(T)$
\begin{align}\label{eq:zcr}
R(T) = \beta_0 + \beta_1\frac{1-\e^{-\gamma T}}{\gamma T}
     + \beta_2\left(\frac{ 1-\e^{-\gamma T}}{\gamma T} - \e^{-\gamma T}\right).
\end{align}
We limit ourselves to two tenors, $x_1$ corresponding to 3 months and $x_2$
corresponding to 6 months. We construct the initial curves from the following
parameters
\begin{center}
\begin{tabular}{c|cccc}
Curve & $\beta_0$ & $\beta_1$ & $\beta_2$ & $\gamma$ \\
\hline
OIS & 0.0003 &   0.01 &   0.07   &    0.06 \\
3m  & 0.0032 &   0.01 &   0.07   &    0.06 \\
6m  & 0.0050 &   0.01 &   0.07   &    0.06 \\
\end{tabular}
\end{center}
In particular, we use \eqref{eq:zcr} to construct the initial 3- and 
6-month LIBOR curves via the expression
\[
L_{k}^{x}(0) = \frac{1}{\delta_{x}} 
	\left( \frac{\exp\left( -R^{x}(T_{k-1})T_{k-1}\right)}
	{\exp\left( -R^{x}(T_{k})T_{k}\right)}-1\right),
\]
for $x=\text{3m and 6m}$, and a third one to construct an initial OIS curve 
consistent with the system
\[
B(0,T_k)=\exp\left(  -R^{OIS}(T_k)T_k\right).
\]
Moreover, we construct the matrices $u^{x_j}$ and $v^{x_j}$ in the following 
simple manner
\begin{align}
u_k^{x_1} &= (u_c \quad \bar{u}^{x_1}_k), \quad k=1,\dots, N^{x_1}\\
u_k^{x_2} &= u_{k\delta_{x_2}/\delta_{x_1}}^{x_1}, \,\,\, \quad k=1,\dots, N^{x_2}\\
v_k^{x_1} &= (\tilde{v}^{x_1}_c \quad  \bar{v}^{x_1}_k),
  \quad k=0,\dots,N^{x_1}-1\\
v_k^{x_2} &= (\tilde{v}^{x_2}_c \quad  \bar{v}^{x_2}_k),
  \quad k=0,\dots,N^{x_2}-1
\end{align}
and $u_{N^{x_1}}^{x_1}=u_{N^{x_2}}^{x_2}=0$, where
$\bar{u}^{x_j}_k,\bar{v}^{x_j}_k\in\Rp$ for $j=1,2$. The bond $B(\cdot,4.5)$
defines the terminal measure, thus $N^{x_1}=18$ and $N^{x_2}=9$.
We set $u_c=0.0065$, $\tilde{v}^{x_1}_c=0.007$ and $\tilde{v}^{x_2}_c=0.0075$.
The remaining values can then be determined uniquely using equations
\eqref{eq:initial-OIS-fit} and \eqref{eq:initial-LIBOR}, i.e. by fitting the
initial term structures. We get that

\begin{center}
\begin{tabular}{l|cccc}
$k$ &$u_k^{x_1}$& $v_k^{x_1}$&$u_k^{x_2}$& $v_k^{x_2}$\\
\hline
0& -&    0.008966 & - &    0.009035 \\
1&0.008638&0.008641&0.008286&0.008358\\
2&0.008286&0.008289&0.007505&0.007577\\
3&0.007908&0.007911&0.006625&0.006697\\
4&0.007505&0.007507&0.005652&0.005725\\
5&0.007077&0.007079&0.004591&0.004664\\
6&0.006625&0.006627&0.003447&0.003520\\
7&0.006150&0.006152&0.002225&0.002298\\
8&0.005652&0.005654&0.000929&0.001003\\
9&0.005132&0.005135 &0 &-\\
10&0.004591&0.004594&&\\
11&0.004029&0.004032&&\\
12&0.003447&0.003450&&\\
13&0.002847&0.002848&&\\
14&0.002225&0.002228&&\\
15&0.001586&0.001589&&\\
16&0.000929&0.000932&&\\
17&0.000254&0.000257&&\\
18&0 &-&&
\end{tabular}
\end{center}
We can observe that all sequences $u^{x_j},v^{x_j}$ for $j=1,2$ are decreasing,
which corresponds to the `normal' market situation; see again Remark
\ref{rem:relation-u-v}.

\subsubsection{Swaption approximation}

Let us consider a 2Y2Y swaption on 3 month LIBOR rates, i.e. a swaption in the
notation of Section \ref{section:Linear_boundary_approximation} with $p=8$ and
$q=16$. We run a Monte Carlo study equivalent to the one in Section
\ref{sec:swaption-ntest} and the results are reported for four different
strikes:

\begin{center}
\begin{tabular}{ccccccc}
Strike ($K$) & $\hat{\mathbb{S}}^+_0$  & Error
& IV (\%) & IV Error &$\mathscr{A}$& $\mathscr{B}$\\
\hline
0.013238&176.17&2.06e-08&30.38&2.326e-10&-5.5403& (1.1596 1)\\
0.023535&52.214&4.31e-08&26.78&1.818e-10&-10.2982& (1.1605 1)\\
0.033831&9.7898&4.09e-08&24.82&2.971e-10&-15.0481& (1.1615 1)\\
0.044128&1.4016&7.90e-09&23.72&2.016e-10&-19.7899& (1.1625 1)\\
\end{tabular}
\end{center}
where
\begin{itemize}
\item $\hat{\mathbb{S}}^+_0:=\hat{\mathbb{S}}^+_0(K,\mathcal{T}^{x_1}_{8,16})$
and IV denote the MC estimator of the price (in basis points) and the implied
volatility (with OIS discounting) using the true exercise boundary defined in
\eqref{eq:f_swaption}.
\item Error
 $:=|\hat{\mathbb{S}}^+_0(K,\mathcal{T}^{x_1}_{8,16}) -
 \hat{\widetilde{\mathbb{S}}}^+_0(K,\mathcal{T}^{x_1}_{8,16})|$, where
 $\hat{\widetilde{\mathbb{S}}}^+_0(K,\mathcal{T}^{x_1}_{8,16})$ denotes the MC
 estimator of the price (in basis points) using the approximate exercise 
 boundary defined in \eqref{eq:ftilde-swaption}.
\item IV Error = $|\text{IV}-\widetilde{\text{IV}}|$, where
 $\widetilde{\text{IV}}$ denotes the implied volatility (with OIS discounting)
 calculated from $\hat{\widetilde{\mathbb{S}}}^+_0(K,\mathcal{T}^{x_1}_{8,16})$.
\item $\mathscr{A}\in\mathbb{R}$ and $\mathscr{B}\in\mathbb{R}^2$ determine the 
 linear approximation to the exercise boundary defined by the function $f$ in 
 \eqref{eq:f_swaption}:
 $$ f(y)\approx \mathscr{A}+\scal{\mathscr{B}}{y}. $$
 Applying the procedure in \citet[pp.~432--434]{SingletonUmantsev02}, we first 
 calculate the upper and lower quantiles for $X_{2}^{(1)}$ using Gaussian 
 approximations for speed. We solve for $x_l$ and $x_u$ in
 \begin{align*}
 f\left( \left[q_{X_{2}^{(1)}}(0.05),{x_l}\right]\right)=0
 \quad\text{ and }\quad 
 f\left( \left[q_{X_{2}^{(1)}}(0.95),{x_u}\right]\right)=0.
 \end{align*}
 Then, $\mathscr{A}$ and $\mathscr{B}$ are computed by fitting the straight line 
 $$\mathscr{A}+\scal{\mathscr{B}}{y}=0$$ through the two points 
 ${y}_l=\left[q_{X_{2}^{(1)}}(0.05),{x}_l\right]$ and 
 ${y}_u=\left[q_{X_{2}^{(1)}}(0.95),{x}_u\right]$. 
\end{itemize}

\subsubsection{Basis swaption approximation}

Let us also consider a 2Y2Y basis swaption. This is an option to enter into a
basis swap paying 3 month LIBOR plus spread $S$ and receiving 6 month
LIBOR, which starts at year 2 and ends at year 4. Once again we conduct a Monte
Carlo study equivalent to Section \ref{sec:basis-swaption-ntest}, and get that

\begin{center}
\begin{tabular}{ccccc}
Spread ($S$) & $\hat{\mathbb{BS}}^+_0$  & Price Error & $\mathscr{C}$&
$\mathscr{D}$\\
\hline
0.0010945&13.778&2.103e-06&-7.7191& (1 5.7514)\\
0.0019458&3.7972&4.784e-05&-14.0029& (1 5.7694)\\
0.0027971&0.64406&9.364e-05&-20.2158& (1 5.7868)\\
0.0036484&0.080951&5.852e-05&-26.3597& (1 5.8037)\\
\end{tabular}
\end{center}
where
\begin{itemize}
\item
 $\hat{\mathbb{BS}}^+_0:=\hat{\mathbb{BS}}^+_0(S,\mathcal{T}^{x_1}_{8,16},
 \mathcal{T}^{x_2}_{4,8})$
 denotes the MC estimator of the price (in basis points) using the true 
 exercise boundary defined in \eqref{eq:f-basis-swaption}.
\item Price Error :=
 $|\hat{\mathbb{BS}}^+_0(S,\mathcal{T}^{x_1}_{8,16},\mathcal{T}^{x_2}_{4,8}) -
 \hat{\widetilde{\mathbb{BS}}}^+_0(S,\mathcal{T}^{x_1}_{8,16},
 \mathcal{T}^{x_2}_{4,8})|$, where similarly
 $\hat{\widetilde{\mathbb{BS}}}^+_0(K,\mathcal{T}^{x_1}_{8,16},
 \mathcal{T}^{x_2}_{4,8})$ denotes the MC estimator of the price (in basis
 points) using the approximate exercise boundary defined in
 \eqref{eq:ftilde-basis-swaption}.
\item $\mathscr{C}\in\mathbb{R}$ and $\mathscr{D}\in\mathbb{R}^2$ determine the 
 linear approximation to the exercise boundary defined by the function $g$ in 
 \eqref{eq:f-basis-swaption}:
 $$ g(y)\approx \mathscr{C}+\scal{\mathscr{D}}{y}.$$
 Applying again the same procedure, we first calculate the upper and lower 
 quantiles for $X_{2}^{(1)}$ and solve for $\tilde{x}_l$ and $\tilde{x}_u$ in
 \begin{align*}
 g\left( \left[q_{X_{2}^{(1)}}(0.05),{\tilde{x}_l}\right]\right)=0
 \quad \text{and } \quad
 g\left( \left[q_{X_{2}^{(1)}}(0.95),{\tilde{x}_u}\right]\right)=0.
 \end{align*}
 Then, $\mathscr{C}$ and $\mathscr{D}$ are computed by fitting the straight 
 line $$\mathscr{C}+\scal{\mathscr{D}}{\tilde{y}}=0$$ through the two points 
 $\tilde{y}_l=\left[q_{X_{2}^{(1)}}(0.05),\tilde{x}_l\right]$ and 
 $\tilde{y}_u=\left[q_{X_{2}^{(1)}}(0.95),\tilde{x}_u\right]$. 
\end{itemize}

These simple examples highlight once again the accuracy of the linear boundary
approximations developed in Sections
\ref{section:Linear_boundary_approximation} and
\ref{section:Linear_boundary_approximation_BS}.
\section{Concluding remarks and future research}
\label{epilogue}

Finally, let us conclude with some remarks that further highlight the merits of 
the affine LIBOR models and some topics for future research. Consider the 
following exotic product: a loan with respect to a 1\$ notional over a monthly 
tenor structure $\mathcal{T}=\{0=T_0 < T_1 < \cdots < T_N\}$ with optional 
interest payments due to the following scheme: At time $t=0,$ the product 
holder may contract to settle the first interest payment either after one, 
three, or six months (as long as the maturity $T_N$ is not exceeded). Next, at 
the first settlement date, the holder may choose again either the one, three, 
or six month LIBOR to be settled one, three or six months later (while not 
exceeding $T_N$). She/He continues until the last payment is settled at $T_N$ 
and the notional is payed back. Clearly, the value of this product at $t=0$ in 
the single curve (pre-crisis) LIBOR world would be simply zero. However, in the 
multiple curve world the pricing of this product is highly non-trivial. In 
particular, such an evaluation would involve the dynamics of any LIBOR rate 
over the periods $\left[T_i,T_j\right],$ $0\leq i\leq j\leq N$ where $T_j-T_i$ 
equals one, three or six months. As a matter of fact, the affine LIBOR model 
with multiple curves presented in this paper is tailor made for this problem as 
it produces `internally consistent' LIBOR and OIS rates over any sub-tenor 
structure. This means that for all sub-tenor structures the rates have the same 
type of dynamics and the driving process remains affine under any forward 
measure. To the best of our knowledge, this is the only multiple curve LIBOR 
model in the literature that naturally produces this consistency across all 
different tenors. The full details of the pricing of this product are, however, 
beyond the scope of this article.

The property of `internal consistency' is beneficial already in the single 
curve LIBOR models. More precisely, the dynamics of LIBOR rates in the 
`classical'  LIBOR market models are specified by setting a `natural' 
volatility structure of a LIBOR system based on a particular tenor structure. 
As a consequence, the volatilities of the LIBOR rates spanning e.g. a double 
period length are immediately hard to determine, as they contain the LIBOR 
rates of the shorter period. On the contrary, in the single  curve affine LIBOR 
models the dynamics of the LIBOR are specified via ratios of martingales that 
are connected  with different underlying tenors, thus one has simultaneously 
specified the dynamics of all possible LIBOR rates in an internally consistent 
way. 

However, the other side of this coin is that a proper choice of the driving 
affine process, and the effective calibration of the affine LIBOR models 
entailed, are far from trivial. In fact these issues require the development of 
new approaches and thus provide a new strand of research on its own. Therefore 
the calibration experiments in this paper are to be considered preliminary and 
merely to demonstrate the potential flexibility of the affine LIBOR model with 
multiple curves.
\appendix
\section{Terminal correlations}
\label{app-corr}

This appendix is devoted to the computation of terminal correlations. The 
expression `terminal correlation' is used in the same sense as in 
\citet[\S6.6]{BrigoMercurio06}, i.e. it summarizes the degree of dependence 
between two LIBOR rates at a fixed, terminal time point. Here the driving 
process is a general affine process and not just an affine 
diffusion as in Section \ref{inst_corr}.

We start by introducing some shorthand notation
\begin{align*}
\Phi_k^x(t) &:= \phi_{T_N-t}(v_{k-1}^x) - \phi_{T_N-t}(u_{k}^x), \\
\Psi_k^x(t) &:= \psi_{T_N-t}(v_{k-1}^x) - \psi_{T_N-t}(u_{k}^x), \\
\Phi_{k_1,k_2}^{x_1,x_2}(t) 
&:= \Phi_{k_1}^{x_1}(t)+\Phi_{k_2}^{x_2}(t), \\
\Psi_{k_1,k_2}^{x_1,x_2}(t) 
&:= \Psi_{k_1}^{x_1}(t)+\Psi_{k_2}^{x_2}(t),
\end{align*}
where $k\in\mathcal{K}^{x}$ and $k_l\in\mathcal{K}^{x_l}$ for $l=1,2$. Then, we 
have from \eqref{eq:LIBOR-rate} that
\begin{align}\label{eq:app-LM}
1 + \delta_{x_l} L_{k_l}^{x_l} (T_i)
&= M_{T_i}^{v_{k_l-1}^{x_l}}/M_{T_i}^{u_{k_l}^{x_l}}
 = \exp\left( \Phi_{k_l}^{x_l}({T_i}) 
       + \bscal{\Psi_{k_l}^{x_l}({T_i})}{X_{T_i}} \right), 
\end{align}
for $l=1,2$ and $T_i\le T_{k_1-1}^{x_1}\vee T_{k_2-1}^{x_2}$. We also denote 
the moment generating function of $X_{T_i}$ under the measure $\P_N$ as follows
\begin{align}
\Theta_{T_i}(z) 
= \E_N\big[\e^{\scal{z}{X_{T_i}}}\big]
= \exp \big(\phi_{T_i}(z) + \scal{\psi_{T_i}(z)}{X_0} \big).
\end{align}
Therefore we get that
\begin{align}
\E_N \left[ M_{T_i}^{v_{k-1}^{x}} / M_{T_i}^{u_{k}^{x}} \right] 
&= \e^{ \Phi_{k}^{x}({T_i}) } 
   \Theta_{T_i}\big(\Psi_{k}^{x}({T_i})\big), \\
\E_N \left[ \left(M_{T_i}^{v_{k-1}^x} / M_{T_i}^{u_k^x}\right)^2\right] 
&= \e^{ 2\Phi_{k}^{x}({T_i}) } 
   \Theta_{T_i}\big(2\Psi_{k}^{x}({T_i})\big), \\
\E_N \left[ M_{T_i}^{v_{k_1-1}^{x_1}} / M_{T_i}^{u_{k_1}^{x_1}} \cdot
           M_{T_i}^{v_{k_2-1}^{x_2}}/ M_{T_i}^{u_{k_2}^{x_2}}\right]
&= \e^{ \Phi_{k_1,k_2}^{x_1,x_2}({T_i})}
   \Theta_{T_i}\big(\Psi_{k_1,k_2}^{x_1,x_2}({T_i})\big). 
\end{align}

The formula for terminal correlations follows after inserting the 
expressions above in the definition of correlation and doing some tedious, but 
straightforward, computations 
\begin{multline*}
\text{Corr}_{T_i}\big[L_{k_1}^{x_1},L_{k_2}^{x_2}\big] 
\stackrel{\eqref{eq:app-LM}}{=}
\text{Corr}\left[ M_{T_i}^{v_{k_1-1}^{x_1}} / M_{T_i}^{u_{k_1}^{x_1}},
              M_{T_i}^{v_{k_2-1}^{x_2}} / M_{T_i}^{u_{k_2}^{x_2}} \right] 
\\= \frac{ \Theta_{T_i}\big(\Psi_{k_1,k_2}^{x_1,x_2}({T_i})\big) 
          - \Theta_{T_i}\big(\Psi_{k_1}^{x_1}({T_i})\big)
           \Theta_{T_i}\big(\Psi_{k_2}^{x_2}({T_i})\big)}
     {\sqrt{\Theta_{T_i}\big(2\Psi_{k_1}^{x_1}({T_i})\big)
        \!-\! \Theta_{T_i}\big(\Psi_{k_1}^{x_1}({T_i})\big)^2}
      \!\sqrt{\Theta_{T_i}\big(2\Psi_{k_2}^{x_2}({T_i})\big)
        \!-\! \Theta_{T_i}\big(\Psi_{k_2}^{x_2}({T_i})\big)^2}}.
\end{multline*}

\bibliographystyle{plainnat}
\bibliography{references}

\end{document}